\documentclass[journal]{IEEEtran}
\usepackage{cite}
\usepackage{amsthm} 
\usepackage{amsmath,amssymb,amsfonts}
\usepackage{algorithmic}
\usepackage{graphicx}
\usepackage{algorithm,algorithmic}
\usepackage{hyperref}
\usepackage{textcomp}
\begin{document}

\title{Prediction-Based Control Barrier Functions for Input-Constrained Safety Critical Systems}
\author{Ali Mesbah, Seid H. Pourtakdoust, Alireza Sharifi, and Afshin Banazadeh
\thanks{Ali Mesbah is with the Department of Aerospace Engineering at the Sharif University of Technology, Tehran, TE, Iran (e-mail: ali.mesbah235@sharif.edu). }
\thanks{Seid H. Pourtakdoust is with the Department of Aerospace Engineering at the Sharif University of Technology, Tehran, TE, Iran (e-mail: pourtak@sharif.edu).}
\thanks{Alireza Sharifi is with the Department of Aerospace Engineering at the Sharif University of Technology, Tehran, TE, Iran (e-mail: ar.sharifi@sharif.edu).}
\thanks{Afshin Banazadeh is with the Department of Aerospace Engineering at the Sharif University of Technology, Tehran, TE, Iran (e-mail: banazadeh@sharif.edu).}
}

\markboth{}
{Mesbah \MakeLowercase{\textit{et al.}}: Prediction-Based Control Barrier Functions}

\newtheorem{theorem}{Theorem}
\newtheorem{proposition}{Proposition}
\newtheorem{definition}{Definition}

\maketitle

\begin{abstract}
Control barrier functions (CBFs) have emerged as a popular topic in safety critical control due to their ability to provide formal safety guarantees for dynamical systems. Despite their powerful capabilities, the determination of feasible CBFs for input-constrained systems is still a formidable task and a challenging research issue. The present work aims to tackle this problem by focusing on an alternative approach towards a generalization of some ideas introduced in the existing CBF literature. The approach provides a rigorous yet straightforward method to define and implement prediction-based control barrier functions for complex dynamical systems to ensure safety with bounded inputs. This is accomplished by introducing a prediction-based term into the CBF that allows for the required margin needed to null the CBF rate of change given the specified input constraints. Having established the theoretical groundwork, certain remarks are subsequently presented with regards to the scheme’s implementation. Finally, the proposed prediction-based control barrier function (PB-CBF) scheme is implemented for two numerical examples. In particular, the second example is related to aircraft stall prevention, which is meant to demonstrate the functionality and capability of the PB-CBFs in handling complex nonlinear dynamical systems via simulations. In both examples, the performance of the PB-CBF is compared with that of a non-prediction based basic CBF.
\end{abstract}


\IEEEpeerreviewmaketitle

\section{Introduction} \label{SecI}
\label{sec:introduction}
\IEEEPARstart{S}{afety} has always been of primary importance and a desired property of any control system. In this sense, remaining within the safe margins is just as important as completing the control task. In other words, a theoretically stable system that results in unsafe behavior is not viable. To this aim, control engineers have often tried to ensure the safety of their design through indirect methods, such as designing for more robustness and stability, so that the control systems can never inadvertently veer off into unsafe states. However, such ad hoc approaches may often be either too conservative or lacking any actual formal guarantees for safety. Further, the concern with assured safety has greatly increased with the rise of artificial intelligence-based control systems that utilize neural networks that practically act as black boxes. Consequently, the subject of control system safety has become a prominent area of active, cutting-edge research within the past decade.

As a recent field in control theory, safety-critical control refers to a series of techniques that aim to generate safe control systems via the development and utilization of rigorous formal methods. In recent years, control barrier functions have grown as one of the topics most synonymous with safety-critical control. Control barrier functions that are very similar in functionality to control Lyapunov functions can formally guarantee safety through invariance inside a safe set.

Similar to Lyapunov stability, the theoretical groundwork for control barrier functions was laid in the twentieth century with one of the earliest commonly cited sources being Mitio Nagumo’s work on set invariance \cite{O1} in the 1940s. The conditions for invariance were independently rediscovered and stated in slightly different forms a number of times, with one of the most notable examples being the work of Jean-Michel Bony \cite{O2} and Haïm Brezis in the 1960s and 1970s, which is known as the Bony-Brazis theorem as yet another popular theorem in set invariance-related work. After the establishment of the theory by such mathematicians, its applications began to manifest in engineering-related control works in the form of “barrier certificates” as a tool to formally guarantee the safe behavior of an existing control system.  The term “control barrier function” was popularized in safety-critical control after being proposed in 2014 by Ames et al., stating that the set invariance conditions could be used in quadratic programming to directly synthesize safe control signals \cite{O3}\cite{O4} resulting in the meteoric rise of safety-critical control in the subsequent years. A more thorough overview of the control barrier functions’ history can be found in \cite{O5}; as such, the remainder of this section is dedicated to the review of the recent relevant literature and its utility.

There are some studies on the notion of adaptiveness for control barrier functions (CBFs) concerning uncertain dynamical systems \cite{O6}\cite{O7}. The application of reinforcement learning to learn and compensate for model uncertainties in the CBF constraints has also been suggested \cite{N8}. In another research, CBFs are extended to “high-order CBFs” in which the Nagumo conditions have to hold consecutively up to a certain high-order derivative resulting in better robustness characteristics \cite{O8}\cite{O9}. In a similar vein, \cite{O10} tackles high-order systems by extending the well-known Lyapunov backstepping method to CBFs by subtracting (instead of adding) the quadratic backstepping term from the initial CBF. Other work has developed the notion of input-to-state safety for the robustness of safety-critical systems subjected to bounded disturbances \cite{N12}. The theory of CBFs has been extended to discrete dynamical systems as well \cite{N13}.

One prominent subject of interest that has concerned a large portion of CBF-related research is the issue of feasibility in the presence of input constraints. In the particular automotive example given by \cite{O3}, the vehicle’s braking distance had to be incorporated into the CBF to ensure that it could be feasibly implemented with acceleration limits. An approach by \cite{O12}, also suggests a framework to obtain feasible solutions based on input-constrained CBFs, which are an extension of high-order barrier functions. As another scheme, \cite{O13} proposes the unification of CBFs with Hamilton-Jacobi reachability to allow for guaranteed bounded disturbance rejection. An approach suggested by \cite{O11} composes feasible CBFs using a Boolean formulation of state constraints. A similar direction is taken by other researchers to construct approximations of CBFs via neural networks. This direction, at the cost of losing some of the hard safety guarantees, allows not only higher scalability with multiple agents involved \cite{O14} but also for the complexity of finding admissible CBFs for input-constrained high-order systems \cite{N18}\cite{O15} to be considerably reduced. The work in \cite{N20} and \cite{N21} propose an optimization-based approach for learning CBFs based on sampled data from demonstrations of how to avoid leaving the safe set. Furthermore, \cite{O16}, \cite{O17}, and \cite{O18} propose an approach to find and remove points in the safe set that become infeasible from the safe set via the introduction of input constraints. 

Finally and most notably, are a number of the ideas brought forward by \cite{O19} and \cite{N26} that may bear some resemblance to the present work. Through what is referred to as future-focused CBFs (FF-CBFs), the safe set is constructed directly based on predicting the future collisions in the system within a finite horizon, if all agents were to follow a zero-control policy \cite{O19}. However, as will be clear soon, FF-CBFs differ from what is proposed especially in how the prediction is performed. Moreover, \cite{N26} formulates the backup strategy approach introduced by \cite{N27} in the CBF framework and extends it to the multi-agent safety problem.

This present study aims to expand upon and extend some of the ideas from the CBF source material to be applicable to an arbitrary input-constrained safety-critical system, through a new approach. In particular, it is suggested that alterations to the safe set due to input constraints may be accounted for via the introduction of an additional term into the CBF to compensate for the changes based on a prediction of the best-case (or near-best-case) scenario to fully stop the advancement toward the safe set’s boundary. It is shown in this study that doing so would guarantee that an admissible control law for the CBF will always exist, thereby solving the feasibility problem. Furthermore, the prediction-based nature of the approach would eliminate the need for problematic calculations such as consecutive differentiation, which greatly facilitates and streamlines the implementation process.

Before explaining the key ideas behind the present study in detail, however, the aforementioned theory behind CBFs shall be overviewed in Section~\ref{SecII}, followed by the particular example that inspired the current work. Having discussed the motivating background, the main contribution of this work referred to as “prediction-based control barrier functions” is introduced in Section~\ref{SecIII}, by presenting a rigorous theoretical foundation on how an existing control barrier function can be augmented with prediction to have feasibility guarantees in an input-constrained system, plus a part dedicated to its implementation. The functionality of prediction-based control barrier functions is then illustrated via two numerical examples in Section~\ref{SecIV}, demonstrating the advantage over the basic CBFs. The first example addresses the part that the prediction-based term plays in refining the super-level set with minimal conservativeness to guarantee the prevention of unsafe behavior when the inputs have constraints. While the second example considers a more detailed and complex aerospace problem that demonstrates the efficacy and applicability of the proposed prediction-based control barrier functions for practical engineering systems. The final Section~\ref{SecV} concludes the present study and discusses future research directives involving prediction-based control barrier functions.

\section{Background} \label{SecII}
In this section, an overview of some of the preliminaries for safety-critical control with control barrier functions (CBFs) is provided, which includes the theoretical background behind set invariance-based safety and CBFs, as well as the most common method for implementing CBF-based safety filters in a control problem. In the final part, a concise overview of the application of CBFs for the problem of adaptive cruise control, as given by \cite{O3} is provided that served as the inspiration for the generalized approach proposed in this study.

By and large, the theory behind CBFs is applicable to a general, nonlinear dynamical system:
\begin{equation}
\mathbf{\dot{x}}=\mathbf{f}(\mathbf{x},\mathbf{u})\text{ ,}
\label{EqO1}
\end{equation}
where $ \mathbf{x} \in \mathbb{R}^n $ and $ \mathbf{u}  \in \mathbb{R}^m $ are the state and input vectors respectively, and the function $ \mathbf{f} : (\mathbb{R}^n , \mathbb{R}^m) \rightarrow \mathbb{R}^n $ is assumed to be locally Lipschitz continuous. If we assume the input $ \mathbf{u} $ to be given by some control law $ \boldsymbol{\pi} : \mathbb{R}^n \rightarrow \mathbb{R}^m $, i.e., $ \mathbf{u} = \boldsymbol{\pi}(\mathbf{x}) $, \eqref{EqO1} becomes:
\begin{equation}
\mathbf{\dot{x}}=\mathbf{f}(\mathbf{x},\boldsymbol{\pi}(\mathbf{x}))=\mathbf{f}(\mathbf{x})\text{ .}
\label{EqO2}
\end{equation}
Furthermore, for the special but fairly common case of a control affine system:
\begin{equation}
\mathbf{\dot{x}}=\mathbf{f}(\mathbf{x})+\mathbf{G}(\mathbf{x})\mathbf{u}\text{ ,}
\label{EqO3}
\end{equation}
with locally Lipschitz continuous functions $ \mathbf{f} : \mathbb{R}^n \rightarrow \mathbb{R}^n $ and $ \mathbf{G} : \mathbb{R}^n \rightarrow \mathbb{R}^{n \times m} $, it is possible to make certain simplifications that would be much more useful and convenient for implementation in safety-critical problems. However, before going over such details, it may be useful to provide a formal definition of safety first.

\subsection{Control Barrier Functions} \label{SecII-A}
One of the most prominent and popular ways to define safety formally is through the notion of set invariance; which itself is defined as:
\begin{definition}[Forward Invariant Set]  \label{DefO1}
A set $ \textit{C} \subseteq \mathbb{R}^n $ is said to be forward invariant for a dynamical system given by \eqref{EqO2} if for every initial condition $ \mathbf{x}_0 \in \textit{C} $, the corresponding solution has the property $ \mathbf{x}(t) \in \textit{C} $ for all $ t \ge 0 $.
\end{definition}
Therefore, if \textit{C} is chosen to be a safe set (i.e., a set of states in which the behavior of the system is considered to be safe), a control law $ \mathbf{u} = \boldsymbol{\pi}(\mathbf{x}) : \mathbb{R}^n \rightarrow \mathbb{R}^m $ that renders the system \eqref{EqO1} forward invariant within the safe set \textit{C} is deemed safe.

As previously mentioned, obtaining such a safe law can be done through a variety of methods, with CBFs being one of the most popular and the method of choice for this work. So, we shall now turn our attention to providing a definition for control barrier functions.

\begin{definition} [Control Barrier Function \cite{O5}]  \label{DefO2}
Let \textit{C} be the super-level set for a continuously differentiable function $ h(\mathbf{x}) : \mathbb{R}^n \rightarrow \mathbb{R} $ with $ \frac{\partial h}{\partial\mathbf{x}}(\mathbf{x}) \neq \mathbf{0} $  when $ h(\mathbf{x}) = 0 $. The function $ h $ is a control barrier function for \eqref{EqO1} on \textit{C} if there exists an extended class $ \text{K}_\infty$ function $ \alpha $ such that for all $ \mathbf{x} \in \textit{C} $ there exists a $ \mathbf{u} \in \mathcal{U} \subseteq \mathbb{R}^m $ such that:
\begin{equation}
\underset{\mathbf{u}\in \mathcal{U}}{\mathop{\sup }}\,\dot{h}(\mathbf{x},\mathbf{u})>-\alpha [h(\mathbf{x})]\text{ .}
\label{EqO4}
\end{equation}
\end{definition}

The $ h(\mathbf{x}) $ given in this definition is also sometimes referred to as a zeroing control barrier function (ZCBF) because of its property $ h(\mathbf{x}) = 0 \;\forall\; \mathbf{x} \in \partial\textit{C} $ as opposed to an older, more Lyapunov-like definition now referred to as reciprocal CBF \cite{O5}, in which the barrier function would blow up on $ \partial\textit{C} $ instead. Still, we shall continue referring to $ h(\mathbf{x}) $ as a CBF in the interest of conciseness. Hence, given a CBF $ h $ for \eqref{EqO1} and a corresponding $ \alpha $, we can define the set of control inputs $ P_\text{CBF}(\mathbf{x}) $ for which the condition \eqref{EqO4} is met \cite{O10} and the following results:

\begin{theorem}[\cite{O5}\cite{O10}]  \label{ThrmO1}
If \textit{C} is the super-level set of a CBF $ h(\mathbf{x}) $ for a system \eqref{EqO1}, then the set $ P_\text{CBF}(\mathbf{x}) $ is non-empty for all $ \mathbf{x} \in \mathbb{R}^n $ and for any locally Lipschitz control law $ \mathbf{u} = \boldsymbol{\pi}(\mathbf{x}) $ with $ \boldsymbol{\pi}(\mathbf{x}) \in P_\text{CBF}(\mathbf{x}) $ for all $ \mathbf{x} \in \mathbb{R}^n $, the resulting \eqref{EqO2} is invariant on \textit{C}.
\end{theorem}

Based on these, one can synthesize controllers that are formally guaranteed to keep the system in question inside a safe set of states which would by extension, also make these laws safe. One may take a number of different approaches to incorporating condition \ref{EqO4} in the control synthesis process, a common method is through quadratic programming (QP) optimization.

\subsection{Implementation with Quadratic Programming Optimization} \label{SecII-B}
The original paper which introduced CBFs as a tool for safe control synthesis, also suggested that the control barrier condition could also be used alongside a control Lyapunov function (CLF) condition in an optimization-based control law \cite{O3}. Excluding the soft CLF constraint, the optimization problem could instead be solved to make a minimal modification to the reference control signal, acting as a safety filter (Fig.~\ref{FigO1}).
\begin{equation}
\begin{array}{*{35}{l}}
{{\mathbf{u}}^{*}}=\mathbf{u}+\underset{\Delta \mathbf{u}\in {{\mathbb{R}}^{m}}}{\mathop{\arg \min }}\,\frac{1}{2}\Delta {{\mathbf{u}}^{\text{T}}}\mathbf{H}\Delta \mathbf{u}  \\
\begin{matrix}
\text{s}\text{.t}\text{.} & \dot{h}(\mathbf{x},\mathbf{u}+\Delta \mathbf{u})+\alpha [h(\mathbf{x})]\ge 0  \\
\end{matrix}  \\
\end{array}
\label{EqO5}
\end{equation}
Here, $ \mathbf{H} \in \mathbb{R}^{m \times m} $ is a positive definite matrix.

As can be seen, \eqref{EqO5} represents a convex optimization problem with nonlinear constraints that is problematic given the fact that in practice, it needs to be solved at each time step. However, the nonlinear constraint is simplified to a linear one if the system is in the control affine form as given by \eqref{EqO3}, allowing \eqref{EqO5} to be solved using QPs.

If there are no limitations to how much the control signal can be modified (i.e., $ \Delta\mathbf{u} \in \mathbb{R}^m $), a well-posed definition of the CBF(s) is sufficient for the optimization to be feasible for all $ \mathbf{x}\in\textit{C} $. However, if additional constraints were to be introduced, namely the input constraints, no such guarantees would exist for the feasibility of the problem. Indeed, it can be shown that an improper selection of the CBF may easily make the QP infeasible \cite{O3}. As such, it is important to define the CBF with consideration of the input constraints. Before introducing the findings of this work, it may be useful to provide some context by going over the specific approach taken by \cite{O3} for the problem of adaptive cruise control to address input constraints.

\begin{figure}[h]
\centerline{\includegraphics[width=\columnwidth]{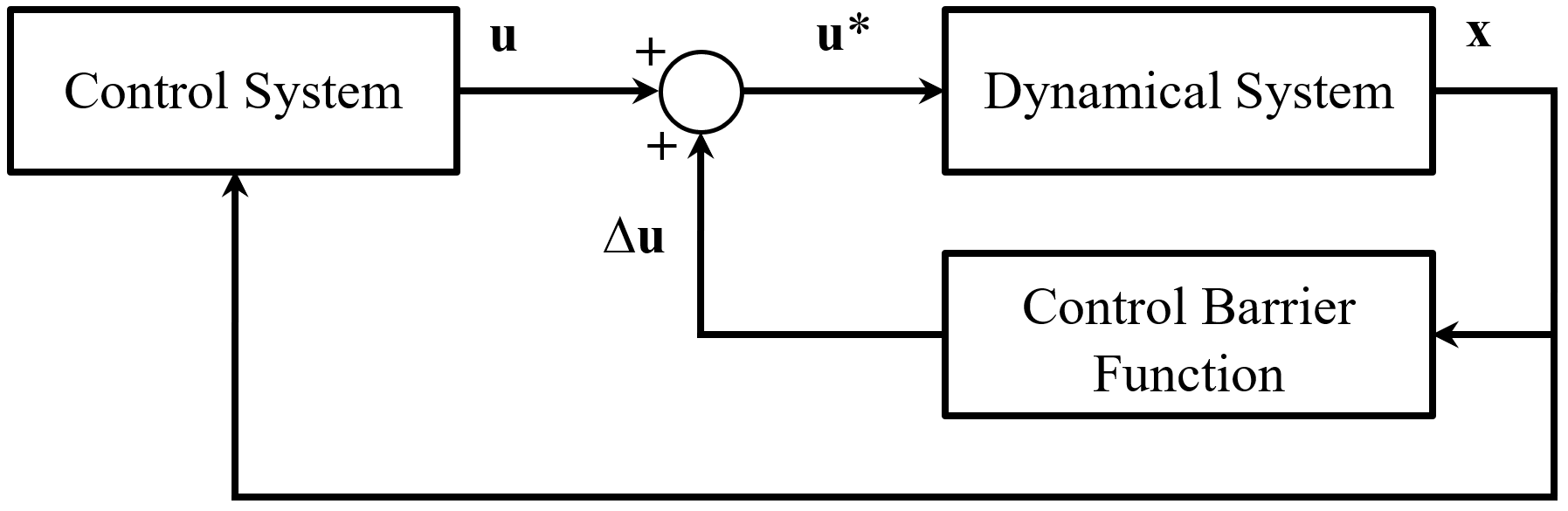}}
\caption{CBF control signal safety filtering.}
\label{FigO1}
\end{figure}

\subsection{Force-Based CBFs for Adaptive Cruise Control}
A more thorough description of adaptive cruise control can be found in \cite{O3} and \cite{O4}; but as the details of this problem are not a key concern of this paper, we shall only go over the essentials that served to inspire the present work.

Adaptive cruise control consists of maintaining an automotive vehicle at a set cruising velocity while avoiding collisions with other vehicles. Considering a simplified single-lane version of the problem, the main vehicle’s motion would be described by the following governing equation:
\begin{equation}
m\frac{dv}{dt}={{F}_{w}}-{{F}_{r}}
\label{EqO6}
\end{equation}
Where $ m $ is the vehicle’s mass, $ v $ is its velocity, $ F_w $ is the wheel force given as the control input, and $ F_r $ is the rolling resistance force that is a function of $ v $. Assuming that the front vehicle is moving at a constant velocity $ v_f $, the distance between the two vehicles will be given by:
\begin{equation}
\frac{dz}{dt}={{v}_{f}}-v
\label{EqO7}
\end{equation}
So, the vehicle’s control system is expected to reach and maintain a set cruise velocity $ v_d $ when possible and slow down to avoid exceeding a minimum safe distance given by $ z_{\min} = \alpha_s v $ (with $ \alpha_s $ being a constant) when appropriate. To this end, a CBF can be defined as:
\begin{equation}
h=z-{{\alpha }_{s}}v
\label{EqO8}
\end{equation}

However, the main vehicle in reality is also limited by how fast it can accelerate and decelerate. That is to say, although the CBF is positive for all $ z > z_{\min} $ and hence indicates safety for all such cases, if the main vehicle were to be going too fast at certain “safe” distances from the front vehicle, even given the maximum possible effort to do so, it would be unable to decelerate to speeds below that of the front vehicle before the minimum safe distance limit is crossed. This is the exact feasibility issue that was alluded to. The approach in \cite{O3} addressed this problem by incorporating into the CBF, the change in relative distance given a full deceleration to match the front vehicle’s velocity.
\begin{equation}
h=z-{{\alpha }_{s}}v-\frac{1}{2}\frac{{{({{v}_{f}}-v)}^{2}}}{{{c}_{d}}g}
\label{EqO9}
\end{equation}
Where $ c_d $ is the maximum possible deceleration rate. By doing this, the safe set corresponding to the CBF will only include those states in which a full (constrained) deceleration before hitting the boundary is possible. This idea can be generalized to apply to any system; which is the exact aim of the present work.

\section{Input-Constrained Control Barrier Functions with Prediction} \label{SecIII}
In many control problems, it may be difficult or even impossible to analytically predict the behavior of the system subject to input constraints. Consequently, for such cases, the governing equations could be used to develop numerical algorithms that can fulfill the same purpose as that of the exact solution while maintaining the formal guarantees at the cost of added computational costs. In this work, one such approach shall be proposed that generalizes the idea of force-based CBFs for adaptive cruise control as presented by \cite{O3} to be applicable for any system where a particular variable is to be kept inside a boundary.

The force-based CBF idea may be extended to the general case through the prediction-based CBF which shall be defined as follows:

\begin{definition}[Prediction-Based Control Barrier Function]  \label{DefO3}
Let $ h(\mathbf{x}) : \mathbb{R}^n \rightarrow \mathbb{R} $ be a CBF on $ \textit{C} \subset \mathbb{R}^n $ for \eqref{EqO1} with the imposed input constraint $ \mathbf{u} \in \mathcal{U}_\text{c} \subset \mathbb{R}^m $. The function $ h_P(\mathbf{x}) = h(\mathbf{x}) + \delta h(\mathbf{x}) $, where
\begin{equation}
\begin{array}{*{35}{l}}
\delta h(\mathbf{x})=\left\{ \begin{array}{*{35}{l}}
\int_{t}^{T}{\dot{h}(\mathbf{x}(\tau ),{{\mathbf{u}}_{0}}(\mathbf{x}(\tau )))d\tau } & (\dot{h}<0)  \\
0 & (\dot{h}\ge 0)  \\
\end{array} \right.  \\
\begin{matrix}
\text{s}\text{.t}\text{.} & \dot{h}(T)=0 & , & \dot{h}(\tau )<0\;\forall\; \tau <T  \\
\end{matrix}  \\
\end{array}
\label{EqO10}
\end{equation}
is a Prediction-Based Control Barrier Function (PB-CBF) on \textit{C} with respect to the constraint $ \mathbf{u} \in \mathcal{U}_\text{c} $, if $ \mathbf{u}_0(\mathbf{x}) : \textit{C} \rightarrow \mathcal{U}_\text{c} $ is a controller under which $ \dot{h}(\mathbf{x},\mathbf{u}_0(\mathbf{x})) $ vanishes in finite time for all $ \mathbf{x} \in \textit{C} $.
\end{definition}

According to this definition, in order to make an existing CBF in compliance with the input constraints, the PB-CBF adds an additional term that accounts for the margin required to halt the advance toward the boundary of \textit{C}. Subsequently, the conditions for input-constrained safety can be presented.
\begin{theorem} \label{ThrmO2}
Given a dynamical system \eqref{EqO1} and the sets \textit{C} and $ \mathcal{U}_\text{c} $ with the corresponding CBF $ h(\mathbf{x}) $ on \textit{C}. If $ h_P(\mathbf{x}) $ is a PB-CBF on \textit{C} with the control law $ \mathbf{u}_0(\mathbf{x}) : \textit{C} \rightarrow \mathcal{U}_\text{c} $, then (\textit{i}) there exists a controller $ \boldsymbol{\pi}(\mathbf{x}) $ such that
\begin{equation}
\dot{h}(\mathbf{x},\boldsymbol{\pi}(\mathbf{x}))-\dot{h}(\mathbf{x},{{\mathbf{u}}_{0}}(\mathbf{x}))+\alpha \left[ {{h}_{P}}(\mathbf{x}) \right]>0 \;,
\label{EqO11}
\end{equation}
which renders the system invariant on a subset of \textit{C} given by:
\begin{equation}
{\textit{C}_\textit{P}}:=\left\{ \mathbf{x}\in {{\mathbb{R}}^{n}}\,\,:\,\,{{h}_{P}}(\mathbf{x})>0 \right\}
\label{EqO12}
\end{equation}
subject to the input constraint $ \boldsymbol{\pi}(\mathbf{x}) \in \mathcal{U}_\text{c} $. (\textit{ii}) Furthermore, if $ \mathbf{u}_0(\mathbf{x}) $ is the optimal control law that minimizes $ \delta h(\mathbf{x}) $ as given by \eqref{EqO10}, then $ \textit{C}_\textit{P} $ will be the largest subset of \textit{C} that can be made invariant given the constraint $ \boldsymbol{\pi}(\mathbf{x}) \in \mathcal{U}_\text{c} \;\forall\; \mathbf{x}\in\textit{C}_\textit{P}$.
\end{theorem}
\begin{proof}
Consider the function $ h_P(\mathbf{x}) $ as given by Definition~\ref{DefO3}; differentiating $ h_P $ yields
\begin{equation}
{{\dot{h}}_{P}}(\mathbf{x})=\dot{h}(\mathbf{x},\boldsymbol{\pi}(\mathbf{x}))+\delta \dot{h}(\mathbf{x})
\label{EqO13}
\end{equation}
Where $ \boldsymbol{\pi}(\mathbf{x}) $ is the control law imposed on the system. Note that $ \delta h $, as defined is continuously differentiable. This is because the integral in \eqref{EqO10} will always be finite under the prescribed control laws, it vanishes as $ \dot{h} \rightarrow 0^+ $ since $ T $ vanishes as $ \dot{h} \rightarrow 0^+ $, and its derivative is also continuous (the last claim will be verified shortly). Therefore, expanding $ \delta \dot{h} $ based on \eqref{EqO10} and applying the Fundamental Theorem of Calculus allows the following conclusion:
\begin{align}
\delta \dot{h}(\mathbf{x})& =\frac{d}{dt}\int\limits_{t}^{T(t)}{\dot{h}(\mathbf{x}(\tau ),{{\mathbf{u}}_{0}})d\tau } \nonumber \\ 
& =\dot{T}\dot{h}(\mathbf{x}(T),{{\mathbf{u}}_{0}})-\dot{h}(\mathbf{x}(t),{{\mathbf{u}}_{0}}) \nonumber \\ 
& =-\dot{h}(\mathbf{x}(t),{{\mathbf{u}}_{0}}) \; .
\label{EqO14}
\end{align}
Note that $ \dot{h}(\mathbf{x}(T),\mathbf{u}_0) = 0 $ by definition. Since $ h $ is differentiable, it follows from \eqref{EqO14} that so is $ \delta h $. Substituting \eqref{EqO14} into \eqref{EqO13} and subsequently applying Nagumo’s theorem leads to the condition \eqref{EqO11} for invariance on $ \textit{C}_\textit{P} $. However, since $ \mathbf{u}_0(\mathbf{x}) \in \mathcal{U}_\text{c} \;\forall\; \mathbf{x}\in\textit{C}_\textit{P} $, selecting $ \boldsymbol{\pi}(\mathbf{x}) = \mathbf{u}_0(\mathbf{x}) $ would allow \eqref{EqO11} to be trivially met since $ \alpha[h_P(\mathbf{x})] > 0 \;\forall\; \mathbf{x} \in \textit{C}_\textit{P} $. This is true while also satisfying the constraint $ \boldsymbol{\pi}(\mathbf{x}) \in \mathcal{U}_\text{c} \;\forall\; \mathbf{x} \in \textit{C}_\textit{P} $. Consequently, we have
\begin{equation}
\sup_{\mathbf{u} \in \mathcal{U}_\text{c}}\{\dot{h}(\mathbf{x},\mathbf{u}) + \alpha[h_P(\mathbf{x})]\} \ge \alpha[h_P(\mathbf{x})] > 0 \;\forall\; \mathbf{x} \in \textit{C}_\textit{P} \;,
\label{EqO15}
\end{equation}
suggesting that the set of admissible controllers for \eqref{EqO11} will always be non-empty at the minimum containing $ \mathbf{u}_0 $. It follows that $ h_P $ is a CBF for $ \textit{C}_\textit{P} $, implying (\textit{i}). The proof for (\textit{ii}) is fairly trivial and will therefore be omitted.
\end{proof}

It may be useful to take a step back and assess the implications of this. In the absence of input constraints, using a basic CBF $ h(\mathbf{x}) $ suffices for enforcing the safety requirements. However, as \cite{O3} has shown for the example of adaptive cruise control and as we shall show for other examples in Section~\ref{SecIV}, with the introduction of input limitations the constraint that can be derived from this CBF may no longer be feasibly enforceable using QPs. That is to say, the activation of the CBF condition may occur too close to the safe barrier, requiring an unattainable control effort. This is evidently caused by the limited ability to bring the rate of change $ \dot{h} $ to zero. So, inspired by \cite{O3}, a prediction term \eqref{EqO10} is defined and incorporated into the existing CBF to account for this limitation to ensure that the new CBF would at the very least work with the control law that the predictions are performed with, thus guaranteeing input-constrained feasibility. Note that the conditions on the control law in \eqref{EqO10} are not very restrictive, as it is not required for the controller to stabilize the system or even $ h $ in particular, but rather to ensure that $ \dot{h} $ vanishes \textit{momentarily} at \textit{some point} in the future. Furthermore, note that $ T $ in \eqref{EqO10} is itself given by
\begin{equation}
T(t)=t+\int\limits_{\dot{h}(t)}^{0}{\frac{d\dot{h}}{{\ddot{h}}}} \;;
\label{EqO16}
\end{equation}
that is not easy to evaluate, much less differentiate for complex high-order dynamical systems. Conveniently however, it was also shown that since $ \dot{h} $ must vanish at $ T $, the problematic task of its explicit evaluation becomes unnecessary for finding the invariance condition.

Based on what has been discussed so far, it is easy to see the distinction between the proposed PB-CBFs and the previously alluded approaches such as FF-CBFs \cite{O19} and backup strategy-based CBFs \cite{N26}. Both of these approaches use a fixed finite time horizon for the predictions, while the PB-CBF approach uses an adaptive time horizon that depends on $ \dot{h} $. However, as previously mentioned, \cite{O19} uses a zero-control policy in the predictions. The approach of \cite{N26} however, does use a control policy, but this control policy is meant to guide the system to a (not necessarily maximal) control invariant set, which in turn requires the determination of such a set. This is in contrast to the control policy used for the PB-CBF method which is only needed to bring the rate $ \dot{h} $ to zero. Finally, \cite{N26} constructs the feasible CBF based on the positive functions defined on the safe and control invariant sets, while a feasible PB-CBF is constructed by adding to an existing unconstrained CBF its total change over the time it takes for $ \dot{h} $ to vanish.

The condition \eqref{EqO11} as given by Theorem~\ref{ThrmO2} is general and can be applied to any Lipschitz continuous nonlinear dynamical system. However, this condition may be simplified if the affine dynamics of \eqref{EqO3} is encountered. Let us define the following for the sake of convenience:
\begin{equation}
\begin{matrix}
\mathbf{c}_{i}^{\text{T}}(\mathbf{x}):=\frac{\partial {{h}_{i}}}{\partial \mathbf{x}} & \text{for any index }i
\end{matrix} \;.
\label{EqO17}
\end{equation}
Hence, we have:
\begin{equation}
\dot{h}(\mathbf{x},\mathbf{u})={{\mathbf{c}}^{\text{T}}}\mathbf{\dot{x}}={{\mathbf{c}}^{\text{T}}}\mathbf{f}(\mathbf{x})+{{\mathbf{c}}^{\text{T}}}\mathbf{G}(\mathbf{x})\mathbf{u} \;;
\label{EqO18}
\end{equation}
and so, \eqref{EqO11} simplifies to:
\begin{equation}
{{\mathbf{c}}^{\text{T}}}\mathbf{G}(\mathbf{x})(\mathbf{u}-{{\mathbf{u}}_{0}})+\alpha \left[ h(\mathbf{x})+\delta h(\mathbf{x}) \right]>0
\label{EqO19}
\end{equation}
It is interesting to note that this condition is no longer dependent on the dynamics, as $ L_\mathbf{f}h = \mathbf{c}^\text{T}\mathbf{f} $ is nowhere to be found, and it only depends on how “far” the control is from the nominal control law.

So to use the QP approach (Section~\ref{SecII-B}) to implement the PB-CBF, $ \delta h $ needs to be evaluated at each time step using numerical methods (or analytically if possible) before simultaneously enforcing \eqref{EqO19} in the optimization problem to guarantee invariance. 

In certain cases, the safety requirements may not be enforceable using a single CBF. An example of this (as it shall be seen in Section~\ref{SecIV}) is when a certain parameter in the system is required to be confined to both an upper and a lower bound as $ z_\text{min} < z < z_\text{max} $. In such cases, the upper and lower requirements may be enforced by defining two separate CBFs that are feasible and building the corresponding PB-CBFs individually. That is, given that the CBFs meet certain criteria as follows:
\begin{proposition} \label{PropO1}
Given the dynamical system of \eqref{EqO3} with the input constraints of $ \mathbf{u} \in \mathcal{U}_\text{c} $ and the functions $ h_{P1}(\mathbf{x}) = h_1(\mathbf{x}) + \delta h_1(\mathbf{x}) $ and $ h_{P2}(\mathbf{x}) = h_2(\mathbf{x}) + \delta h_2(\mathbf{x}) $, if these functions are individually PB-CBFs on the sets $ \textit{C}_1 $ and $ \textit{C}_2 $ respectively, and it holds that $ \mathbf{c}_1~/~\|\mathbf{c}_1\| = -\mathbf{c}_2~/~\|\mathbf{c}_2\| $ for all $ \mathbf{x} \in \textit{C}_1 \cap \textit{C}_2 $, then $ h_{P1} $ and $ h_{P2} $ are simultaneously PB-CBFs on $ \textit{C}_1 \cap \textit{C}_2 $ with respect to $ \mathbf{u} \in \mathcal{U}_\text{c} $.
\end{proposition}
\begin{proof}
For all $ \mathbf{x} \in \mathbb{R}^n $ such that $ \sup_{\mathbf{u} \in \mathcal{U}_\text{c}}\dot{h}_1(\mathbf{x},\mathbf{u}) \leq 0 $, it follows that
\begin{align}
\begin{split}
\underset{\mathbf{u}\in {{\mathcal{U}}_{\text{c}}}}{\mathop{\inf }}\,{{{\dot{h}}}_{2}}(\mathbf{x},\mathbf{u})& =\underset{\mathbf{u}\in {{\mathcal{U}}_{\text{c}}}}{\mathop{\inf }}\,\mathbf{c}_{2}^{\text{T}}(\mathbf{x})\mathbf{f}(\mathbf{x},\mathbf{u}) \\
& =\underset{\mathbf{u}\in {{\mathcal{U}}_{\text{c}}}}{\mathop{\inf }}\left\{-\frac{||{{\mathbf{c}}_{2}}||}{||{{\mathbf{c}}_{1}}||}\mathbf{c}_{1}^{\text{T}}(\mathbf{x})\mathbf{f}(\mathbf{x},\mathbf{u})\right\} \\ 
& =\frac{||{{\mathbf{c}}_{2}}||}{||{{\mathbf{c}}_{1}}||}\underset{\mathbf{u}\in {{\mathcal{U}}_{\text{c}}}}{\mathop{\inf }}\left\{-{{{\dot{h}}}_{1}}(\mathbf{x},\mathbf{u})\right\} \\
& =-\frac{||{{\mathbf{c}}_{2}}||}{||{{\mathbf{c}}_{1}}||}\underset{\mathbf{u}\in {{\mathcal{U}}_{\text{c}}}}{\mathop{\sup }}\,{{{\dot{h}}}_{1}}(\mathbf{x},\mathbf{u})\ge 0 \;.    
\end{split}
\label{EqO20}
\end{align}
Hence, for any admissible control $ \boldsymbol{\pi}_1(\mathbf{x}) \in \mathcal{U}_\text{c} $ for the PB-CBF $ h_{P1} $, we have $ \delta h_2 = 0 $; so, assuming that the system is inside the intersection of the safe sets $ \textit{C}_1 \cap \textit{C}_2 $, it is obvious that $ h_{P2}(\mathbf{x}) > 0 $ (if $ \mathbf{x} \in \textit{C}_2 - \textit{C}_1 $, then $ h_1 < 0 $ which violates the CBF condition on $ \textit{C}_1 $). Consequently,
\begin{align}
\begin{split}
{{\dot{h}}_{P2}}(\mathbf{x},{{\boldsymbol{\pi}}_{1}}(\mathbf{x}))& +{{\alpha }_{2}}[{{h}_{P2}}(\mathbf{x})] \\
& ={{\dot{h}}_{2}}(\mathbf{x},{{\boldsymbol{\pi}}_{1}}(\mathbf{x}))+{{\alpha }_{2}}[{{h}_{P2}}(\mathbf{x})]>0 \;.    
\end{split}
\label{EqO21}
\end{align}
Thus, the invariance condition for $ h_{P2}$ is automatically met for any such controller. We can show in a similar way that if $ \sup_{\mathbf{u} \in \mathcal{U}_\text{c}}\dot{h}_2(\mathbf{x},\mathbf{u}) \leq 0 $, it is sufficient for $ \mathbf{u} $ to be selected such that the $ \textit{C}_2 $ PB-CBF condition is satisfied, with the $ \textit{C}_1 $ condition being met automatically. Therefore, both PB-CBF conditions are simultaneously feasible on $ \textit{C}_1 \cap \textit{C}_2 $.
\end{proof}

The proof of Proposition~\ref{PropO1} implies that at any given time, based on the sign of $ \dot{h}_1 $ and $ \dot{h}_2 $, only one of the PB-CBFs is “active” since it is the direction of change for $ h $ that determines which safety barrier is at risk of being crossed. This in turn, means that at any time, only the condition corresponding to the active PB-CBF needs to be checked and met.

\subsection{Control Law Choice in the Prediction} \label{SecIII-A}

So far, the focus of this paper has been mostly limited to the feasibility of PB-CBFs in guaranteeing input-constrained set invariance through the addition of the prediction-based term; however, an essential question is yet to be answered: how should the control law in the prediction as seen in \eqref{EqO10} and \eqref{EqO11} should be selected?

As Definition~\ref{DefO3} outlines, the control law $ \mathbf{u}_0(\mathbf{x}) $ used in the prediction is arbitrary as long as it provides the guarantee that the integral in \eqref{EqO10} is always bounded in the safe set. But of course, a sloppy choice of $ \mathbf{u}_0 $ can result in the invariant subset becoming overly restricted or in the extreme case, altogether empty! So, getting a desirable performance from the PB-CBF is strongly dependent on a competent selection of $ \mathbf{u}_0 $. On the other hand, though it is the best option from the performance standpoint, finding the control law that is described by part (\textit{ii}) of Theorem~\ref{ThrmO2} is no trivial task. Generally speaking, this would entail solving a nonlinear optimal control problem to maximize $ \delta h $ over a finite time horizon. What is more, this problem would need to be solved again \textit{at every time step}, which would make it extremely computationally intensive for a real-time implementation or highly unscalable if approached with tabular methods.

Consequently, in many cases, it may be the best choice to settle for a suboptimal solution that comes reasonably close to the truly optimal law but is less computationally cumbersome. Furthermore, it may even be desirable not to go for the most optimal choice of $ \mathbf{u}_0 $ in order to make the predictions slightly more conservative, thus improving robustness (a topic that is not addressed in this work, but one that may be interesting to look into in the future). The rest of this section will therefore be dedicated to providing some suggestions for what $ \mathbf{u}_0 $ could be.

One probable choice instead of finding the $ \mathbf{u}_0(\tau) $ that minimizes $ \delta h $ over the interval $ \tau \in [t\;,\;T] $, could be to find $ \mathbf{u}_0 $ such that the rate of reduction of $ \dot{h} $ is instantaneously maximized (thus minimizing the length of the interval itself); that is:
\begin{equation}
{{\mathbf{u}}_{0}}(\tau )=\underset{\mathbf{u}\in {{\mathcal{U}}_{\text{c}}}}{\mathop{\arg \max }}\,\ddot{h}(\mathbf{x}(\tau ),\mathbf{u}) \;.
\label{EqO22}
\end{equation}
Where $ \ddot{h} $ is found by differentiating \eqref{EqO18} and substituting \eqref{EqO3}, as given by \eqref{EqO23}.
\begin{align}
\begin{split}
\ddot{h}(\mathbf{x},\mathbf{u})& ={{{\mathbf{\dot{x}}}}^{\text{T}}}{{\left( \frac{\partial \mathbf{c}}{\partial \mathbf{x}} \right)}^{\text{T}}}(\mathbf{f}+\mathbf{Gu})+{{\mathbf{c}}^{\text{T}}}\left( \frac{\partial \mathbf{f}}{\partial \mathbf{x}}\mathbf{\dot{x}}+\frac{\partial \mathbf{G}}{\partial \mathbf{x}}\mathbf{\dot{x}u} \right) \\ 
& ={{\mathbf{f}}^{\text{T}}}{{\left( \frac{\partial \mathbf{c}}{\partial \mathbf{x}} \right)}^{\text{T}}}\mathbf{f}+{{\mathbf{c}}^{\text{T}}}\frac{\partial \mathbf{f}}{\partial \mathbf{x}}\mathbf{f}+{{\mathbf{c}}^{\text{T}}}\frac{\partial \mathbf{f}}{\partial \mathbf{x}}\mathbf{Gu} \\ 
& +{{\mathbf{c}}^{\text{T}}}\left( \frac{\partial \mathbf{G}}{\partial \mathbf{x}}\mathbf{f} \right)\mathbf{u}+{{\mathbf{f}}^{\text{T}}}\left[ {{\left( \frac{\partial \mathbf{c}}{\partial \mathbf{x}} \right)}^{\text{T}}}+\left( \frac{\partial \mathbf{c}}{\partial \mathbf{x}} \right) \right]\mathbf{Gu} \\ 
& +{{\mathbf{c}}^{\text{T}}}\left( \frac{\partial \mathbf{G}}{\partial \mathbf{x}}\mathbf{Gu} \right)\mathbf{u}+{{\mathbf{u}}^{\text{T}}}{{\mathbf{G}}^{\text{T}}}{{\left( \frac{\partial \mathbf{c}}{\partial \mathbf{x}} \right)}^{\text{T}}}\mathbf{Gu}
\end{split}
\label{EqO23}
\end{align}
Subsequently, defining the mapping $ \mathbf{M}(\mathbf{x}) $ such that
\begin{equation}
{{\mathbf{c}}^{\text{T}}}\left( \frac{\partial \mathbf{G}}{\partial \mathbf{x}}\mathbf{Gu} \right)\mathbf{u}={{\mathbf{u}}^{\text{T}}}\mathbf{M}(\mathbf{x})\mathbf{u}
\label{EqO24}
\end{equation}
Allows for \eqref{EqO23} to be rewritten in the quadratic format.
\begin{align}
\begin{split}
&\ddot{h}(\mathbf{x},\mathbf{u})={{\mathbf{u}}^{\text{T}}}\left[ {{\mathbf{G}}^{\text{T}}}{{\left( \frac{\partial \mathbf{c}}{\partial \mathbf{x}} \right)}^{\text{T}}}\mathbf{G}+\mathbf{M}(\mathbf{x}) \right]\mathbf{u} \\
&+\left\{ {{\mathbf{f}}^{\text{T}}}\left[ {{\left( \frac{\partial \mathbf{c}}{\partial \mathbf{x}} \right)}^{\text{T}}}+\left( \frac{\partial \mathbf{c}}{\partial \mathbf{x}} \right) \right]\mathbf{G}+{{\mathbf{c}}^{\text{T}}}\frac{\partial \mathbf{f}}{\partial \mathbf{x}}\mathbf{G}+{{\mathbf{c}}^{\text{T}}}\left( \frac{\partial \mathbf{G}}{\partial \mathbf{x}}\mathbf{f} \right) \right\}\mathbf{u}
\end{split}
\label{EqO25}
\end{align}
Thus, if the constraints on $ \mathbf{u}_0 $ are all linear, i.e.,
\begin{equation}
{{\mathbf{u}}_{\text{min}}}\le \mathbf{u}\le {{\mathbf{u}}_{\text{max}}}\;,
\label{EqO26}
\end{equation}
then at each step of the prediction process, \eqref{EqO22} may be solved using quadratic programming according to \eqref{EqO25} to find $ \mathbf{u}_0 $.

If the system’s controls behave in a sufficiently linear-like fashion (i.e., controls are monotonic with no reversals), the process of solving \eqref{EqO22} can be further simplified through linear approximation.
\begin{equation}
\mathbf{f}(\mathbf{x})\approx {{\left. \frac{\partial \mathbf{f}}{\partial \mathbf{x}} \right|}_{{{\mathbf{x}}_{0}}}}(\mathbf{x}-{{\mathbf{x}}_{0}})=\mathbf{A}(\mathbf{x}-{{\mathbf{x}}_{0}})
\label{EqO27}
\end{equation}
\begin{equation}
\mathbf{B}\approx \mathbf{G}({{\mathbf{x}}_{0}})
\label{EqO28}
\end{equation}
\begin{equation}
{{\mathbf{c}}^{\text{T}}}\approx {{\mathbf{c}}^{\text{T}}}({{\mathbf{x}}_{0}})
\label{EqO29}
\end{equation}
Where $ \mathbf{x}_0 $ is some state representative of the general trend of the controls’ behavior (note that this state does not have to be fixed). Substituting \eqref{EqO27}, \eqref{EqO28} and \eqref{EqO29} into \eqref{EqO23} eliminates its dependency on u (as can be seen in \eqref{EqO30}).
\begin{equation}
\frac{\partial \ddot{h}}{\partial \mathbf{x}}(\mathbf{x})={{\mathbf{c}}^{\text{T}}}\mathbf{AB}
\label{EqO30}
\end{equation}
Thus, under the linear-like controls assumption, the solution to \eqref{EqO25} corresponds to a bang-bang law \cite{O22}.
\begin{equation}
{{u}_{i}}=\left\{\begin{array}{*{35}{l}}
{{u}_{i,\max }} & (\pm \text{sgn} [{{\mathbf{c}}^{\text{T}}}{{\mathbf{B}}_{i}}]>0)  \\
{{u}_{i,0}} & (\pm \text{sgn} [{{\mathbf{c}}^{\text{T}}}{{\mathbf{B}}_{i}}]=0)  \\
{{u}_{i,\min }} & (\pm \text{sgn} [{{\mathbf{c}}^{\text{T}}}{{\mathbf{B}}_{i}}]<0)  \\
\end{array} \right.
\label{EqO31}
\end{equation}

Note that even if the controls are not monotonic, but reversals do not occur, the near-optimal solution given by \eqref{EqO31} would still serve as a suitable, computationally inexpensive alternative to solving \eqref{EqO22} subject to \eqref{EqO25} resulting in a slightly more conservative $ \delta h $.

\subsection{Implementation} \label{SecIII-B}
In this section, the results from previous sections shall be compiled to provide the reader with a procedure to implement the PB-CBFs in a control problem. Let us consider the affine system of \eqref{EqO3} with the input constraints $ \mathbf{u}_{\min} \leq \mathbf{u} \leq \mathbf{u}_{\max} $ and the CBF $ h(\mathbf{x}) $. Assuming that the input $ \mathbf{u} $ is given by some arbitrary Lipschitz continuous controller, the PB-CBF must make minimal admissible modifications to $ \mathbf{u} $ in order to maintain the system inside the safe region \eqref{EqO12}. To this end, the PB-CBF can be integrated into the control scheme as shown in Fig.~\ref{FigO1}. The PB-CBF in this diagram, as shown in Fig.~\ref{FigO2}, is itself comprised of two parts: the portion that estimates the stopping margin $ \delta h $ using model-based prediction/propagation, and the portion that makes the feasible modification on $ \mathbf{u} $ based on the CBF that incorporates $ \delta h $.
\begin{figure}[h]
\centerline{\includegraphics[width=\columnwidth]{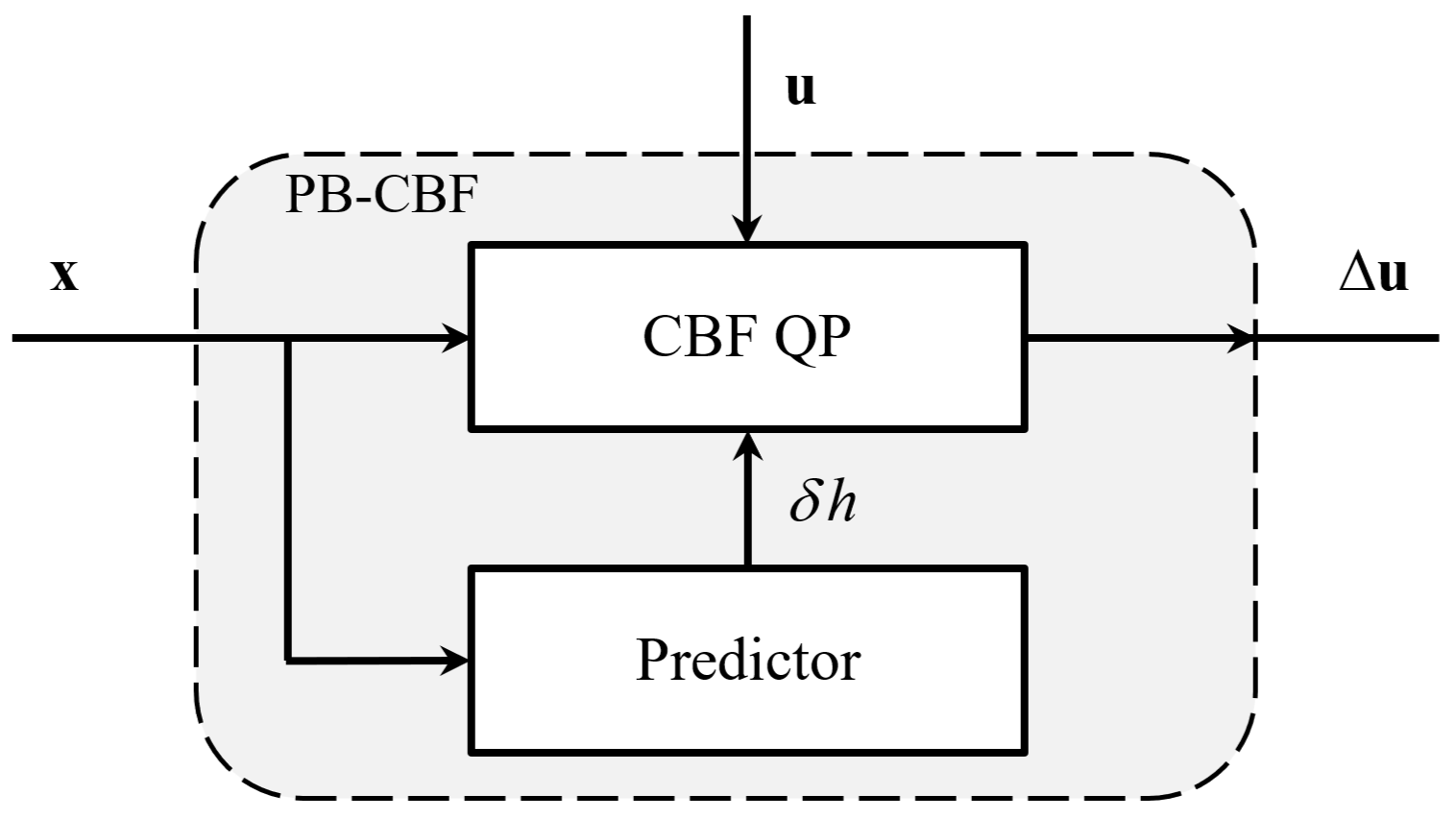}}
\caption{Modules of the prediction-based control barrier function.}
\label{FigO2}
\end{figure}
The scheme presented in Fig.~\ref{FigO2} may be implemented based on the following pseudocode.
\begin{algorithm}[H]
\caption{Prediction-Based Control Barrier Function.}\label{alg:alg1}
\begin{algorithmic}
\STATE
\STATE Input integration time step $ \Delta t $
\STATE \textbf{do} at every time step:
\STATE \hspace{0.5cm}\textbf{if} $ \dot{h}(\mathbf{x},{{\mathbf{u}}_{0}})<0 $ then
\STATE \hspace{0.75cm}Propagate \eqref{EqO3} from the present state with $ \mathbf{u} = \mathbf{u}_0(\mathbf{x}) $ until $ T $ when $ \dot{h}(\mathbf{x},{{\mathbf{u}}_{0}})\ge 0 $ to find $ \mathbf{x}(\tau) $ for $ \tau \in [t\;,\;T] $
\STATE \hspace{0.75cm}Calculate $ \delta h $ using \eqref{EqO10}
\STATE \hspace{0.5cm}\textbf{else}
\STATE \hspace{0.75cm}Set $ \delta h = 0 $
\STATE \hspace{0.5cm}\textbf{end if}
\STATE \hspace{0.5cm}Solve
\STATE \hspace{0.75cm} \begin{align*}
  & \Delta \mathbf{u}=\underset{\Delta \mathbf{u}\in {{\mathbb{R}}^{m}}}{\mathop{\arg \min \frac{1}{2}\Delta {{\mathbf{u}}^{\text{T}}}\mathbf{H}\Delta \mathbf{u}}}\, \\ 
 & \begin{array}{*{35}{l}}
   \text{s.t.} & \begin{split}
   -{{\mathbf{c}}^{\text{T}}}\mathbf{G}(\mathbf{x})\Delta \mathbf{u}&<{{\mathbf{c}}^{\text{T}}}\mathbf{G}(\mathbf{x})(\mathbf{u}-{{\mathbf{u}}_{0}}) \\
   &+\alpha \left[ h(\mathbf{x})+\delta h(\mathbf{x}) \right]
   \end{split} \\
   {} & \Delta \mathbf{u}<{{\mathbf{u}}_{\max }}-\mathbf{u}  \\
   {} & -\Delta \mathbf{u}<\mathbf{u}-{{\mathbf{u}}_{\min }}  \\
\end{array}
\end{align*}
\STATE \textbf{end do}
\end{algorithmic}
\label{AlgO1}
\end{algorithm}

As discussed, $ \mathbf{u}_0(\mathbf{x}) $ may be found either by solving \eqref{EqO22} or by using \eqref{EqO31} if the required conditions are met. It should be noted that, since the estimated stopping margin $ \delta h $ is numerically evaluated, a sufficiently small timestep $ \Delta t $ must be selected to prevent the upper bound of $ \delta h $ from exceeding the value of its theoretically possible maximum $ \delta h^\ast $ (based on the optimal solution). Otherwise, running into feasibility problems at some (albeit, fewer than usual) values of $ \mathbf{x} $ is possible. As such, for the cases where computational limitations are prominent, an interesting direction for future work may be to look into numerical methods in the propagation and integration process that would allow for the upper bound of $ \delta h $ to always be smaller than that of $ \delta h^\ast $, regardless of the value of $ \Delta t $.

\section{Numerical Examples} \label{SecIV}

So far, the reader has been provided with a framework for how PB-CBFs may be implemented. In order to both illustrate the necessity of the prediction-based term and further clarify the process of building PB-CBFs from existing CBFs, this section shall be dedicated to implementing PB-CBFs in some numerical examples.

\subsection{Double Integrator} \label{SecIV-A}

Before exploring a more complicated example, it may be useful to implement the PB-CBF scheme in a simple problem to hopefully gain some intuition into its functionality. Consider a 2D double integrator system $ \dot{\mathbf{r}} = \mathbf{v} $ and $ \dot{\mathbf{v}} = \mathbf{u} $, where $ \mathbf{r} $ is the position, $ \mathbf{v} $ is the velocity, and $ \mathbf{u} $ is the input acceleration. The following input constraint is imposed on the system.
\begin{equation}
\| u \|_2 \leq u_{\max}
\label{EqO32}
\end{equation}
Furthermore, the system is required to be kept outside the circle $ \left\{ r \in \mathbb{R}^2 : \|r\|_2 \leq R \right\} $ where $ R $ is the radius of the circle.

Given the geometry of the safe set’s boundaries, a polar coordinates representation of the system equations as given by \eqref{EqO33} is the best choice for the CBF calculations.
\begin{equation}
\left\{ \begin{array}{*{35}{l}}
\dot{r}={{v}_{r}}  \\
\dot{\theta }={{\omega }_{\theta }}  \\
\ddot{r}={{u}_{r}}+r{{\omega }_{\theta }}^{2}  \\
\ddot{\theta }=\frac{{{u}_{\theta }}}{r}-\frac{2{{v}_{r}}{{\omega }_{\theta }}}{r}  \\
\end{array} \right.
\label{EqO33}
\end{equation}
Defining the state, and input vectors $ x = [r\;\theta\;\dot{r} = v_r\;\dot{\theta} = \omega_\theta]^\text{T} $ and $ \mathbf{u} = [u_r  u_\theta]^\text{T} $, the safe set for this system will be given by $ \textit{C}_0 := \left\{\mathbf{x} \in \mathbb{R}^4 : r > R^2\right\} $. Ignoring the input constraint for now, using the backstepping method introduced by \cite{O10}, the following CBF can be defined to render the unconstrained system invariant inside $ \textit{C} \subset \textit{C}_0 $.
\begin{equation}
h(\mathbf{x})=r-R-\frac{1}{2\mu }{{\dot{r}}^{2}}
\label{EqO34}
\end{equation}

The CBF \eqref{EqO34} would produce admissible results for all $ \mu \in \mathbb{R}^+ $ if no input constraints were present; however, as shall be shown shortly, this will not be the case when \eqref{EqO32} is imposed. The next step is to build the PB-CBF from \eqref{EqO34} and a reference control law $ \mathbf{u}_0(\mathbf{x}) $. Given the nature of the input constraint, a prudent choice for $ \mathbf{u}_0 $ may be the radial version of \eqref{EqO31},
\begin{equation}
{{\mathbf{u}}_{0}}(\mathbf{x})={{u}_{\max }}\frac{{{\mathbf{c}}^{\text{T}}}\mathbf{B}}{||{{\mathbf{c}}^{\text{T}}}\mathbf{B}||}
\label{EqO35}
\end{equation}
where $ \mathbf{c}^\text{T} $ and $ \mathbf{B} $ are calculated using \eqref{EqO28} and \eqref{EqO29} respectively (with $ \mathbf{x}_0 = \mathbf{x} $).
\begin{equation}
\begin{matrix}
{{\mathbf{B}}^{\text{T}}}=\left[ \begin{array}{*{35}{l}}
0 & 0 & 1 & 0  \\
0 & 0 & 0 & {1}/{r}  \\
\end{array} \right] & , & {{\mathbf{c}}^{\text{T}}}=[\begin{matrix}
1 & 0 & {-\dot{r}}/{\mu } & 0  \\
\end{matrix}]  \\
\end{matrix}
\label{EqO36}
\end{equation}
The control law \eqref{EqO35} is thus simplified to the expression \eqref{EqO37}.
\begin{equation}
{{\mathbf{u}}_{0}}(\mathbf{x})=-{{u}_{\max }}{{[\begin{matrix}
1 & 0  \\
\end{matrix}]}^{\text{T}}}\text{sgn} (\dot{r})
\label{EqO37}
\end{equation}
Hence, using \eqref{EqO10}, \eqref{EqO33} and \eqref{EqO37}, $ \delta h $ could be calculated based on the current state to construct the PB-CBF. However, before proceeding to the simulation results it may be useful to further look into how $ \delta h $ would turn out based on different values of $ u_{\max} $ and $ \mu $. So, consider the time derivative of $ h $:
\begin{equation}
\dot{h}(\mathbf{x},\mathbf{u})={{\mu }^{-1}}\dot{r}(\mu -{{u}_{r}}-r{{\dot{\theta }}^{2}})
\label{EqO38}
\end{equation}
Based on \eqref{EqO38}, there are two possible scenarios. If $ u_{\max} < \mu - r\dot{\theta}^2 $, then the sign of $ \dot{h} $ would be the same as $ \dot{r} $ regardless of what $ \mathbf{u} $ is chosen to be. That is to say, $ \dot{h} $ cannot be instantaneously brought to zero solely through the right choice of $ \mathbf{u} $; and so, for all $ \dot{r} < 0 $ it takes a nonzero amount of time to bring $ \dot{h} $ to zero, resulting $ \delta h $ being negative for all $ \dot{r} < 0 $, implying that $ \textit{C}_\textit{P} \subset \textit{C} $. Furthermore, under $ \mathbf{u}_0 $, $ \dot{h} $ is guaranteed to vanish in finite time since,
\begin{equation}
\ddot{h}(\mathbf{x},{{\mathbf{u}}_{0}})={{\mu }^{-1}}({{u}_{\max }}+r{{\dot{\theta }}^{2}})(\mu -{{u}_{\max }}-r{{\dot{\theta }}^{2}})+{{\dot{r}}^{2}}{{\dot{\theta }}^{2}}>0
\label{EqO39}
\end{equation}
implies that it is strictly increasing. The other case where $ \dot{r} $ is positive and $ \dot{\theta} $ is large, shall not be discussed, but a similar reasoning could be followed to make the same conclusions. On the other hand, if $ u_{\max} \geq \mu - r\dot{\theta}^2 $ (and $ r\dot{\theta}^2 \leq \mu + u_{\max} $), then the expression \eqref{EqO38}, can be made to vanish (or be positive) instantaneously, which is what the selected $ \mathbf{u}_0 $ does, resulting in $ \delta h = 0 $.

Due to the presence of the centripetal term in \eqref{EqO38}, there is no value of $ \mu $ that can be definitively said to always make $ h $ feasible for a given $ u_{\max} $. However, in order for the CBF \eqref{EqO34} to be as unrestrictive as possible, the best choice is to select $ \mu = u_{\max} $. This, at least, ensures that if $ \dot{\theta} $ is sufficiently small, the CBF will always be feasible. In many problems, however, even this vague conclusion may be hard to come across. In other words, it would almost certainly be far more difficult than finding a control law that brings $ \dot{h} $ to zero. Nevertheless, the PB-CBF constructed from a CBF with $ \mu > u_{\max} $ would still be less restrictive! Consider Fig.~\ref{FigN3} which shows a $ \theta = 0 $ and $ \omega_\theta = 0 $ cross-section of the sets $ \textit{C}_0 $, \textit{C} and $ \textit{C}_\textit{P} $ in the $ r\text{–}\dot{r} $ plane. As can be seen, $ \textit{C}_\textit{P} $ is only missing a section from the bottom side of the \textit{C} where the values of $ \dot{r} $ are negative. This is in contrast to sections being removed from both the bottom and top (unnecessarily so for the latter) if the other option was taken. Therefore, another advantage of the PB-CBF is the fact that it is asymmetric, which allows it to be less restrictive.
\begin{figure}[h]
\centerline{\includegraphics[width=0.8\columnwidth]{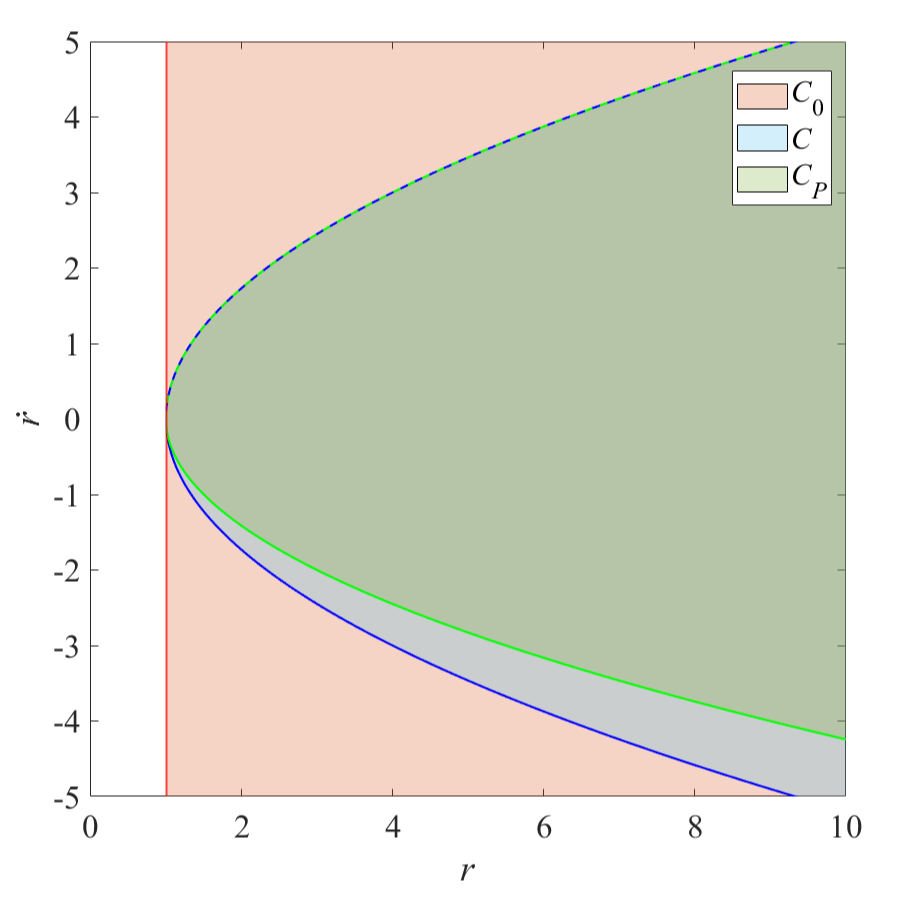}}
\caption{Graph of the safe set, the CBF super-level set, and the PB-CBF super-level set in the $ r\text{–}\dot{r} $ phase plane for $ \mu = 1.5 $ and $ u_{\max} = 1 $.}
\label{FigN3}
\end{figure}

Similar to the parameter $ \mu $ in the CBF \eqref{EqO34}, an appropriate selection of the class $ \text{K}_\infty$ function $ \alpha $ in \eqref{EqO4} may still result in solutions to \eqref{EqO5} that are admissible with the given input constraints. This in general also faces the same sort of issues as seen with the selection of the CBF itself, in that not only may it be difficult to find the appropriate class $ \text{K}_\infty$ function to begin with, but also, the corresponding results may be overly conservative. In order to illustrate these claims, the simulation results of the system with the presented CBF and the respective PB-CBF shall be examined.

The system with $ u_{\max} = 1 $ is initiated from $ \mathbf{r}(0) = [-2\; 1]^\text{T} $ and $ \dot{\mathbf{r}}(0) = [1\;0.25]^\text{T} $ for which it is desired to follow the reference path $ \mathbf{r}_r(t) = [\sin(t\pi/5)\;t\pi/10]^\text{T} $ (both vectors are represented in the Cartesian coordinates $ \mathbf{r} = [x\;y]^\text{T} = r[\cos\theta\; \sin\theta]T $) while always remaining outside of the unsafe region. This was achieved by designing a state feedback controller to track the reference, and then integrating the CBFs into the scheme. Furthermore, if the CBF is unable to maintain the system outside of the unsafe region, a secondary mechanism is designed to push it back into the safe set. Choosing $ \mu = 1.5 $ and $ \alpha (z) = \gamma z $ where $ \gamma > 0 $, Fig.~\ref{FigN4} and Fig.~\ref{FigN4} show the simulation results for $ \gamma = 2 $ and $ \gamma = 10 $ respectively.

Based on the results, two observations can be made. In the case where $ \gamma = 10 $, while the PB-CBF is able to keep the system from crossing the safe boundary, the base (prediction-less) CBF is unable to do this. The effect of the prediction term allows the PBCBF to “react” slightly (about 35 ms) sooner resulting in the system “turning” away from the safe boundary earlier and therefore not crossing it. This is of course consistent with the previously obtained theoretical results that claimed the system can be feasibly made invariant with the PB-CBF using any class $ \text{K}_\infty$ function. Another interesting observation is also made from the case where $ \gamma = 2 $. In this case, the base CBF is also able to keep the system from crossing the boundary (though this may not be true for all possible initial conditions); however, in contrast to the other case where the PB-CBF was activated sooner than the base CBF, here it can be seen that the PB-CBF is activated slightly (about 46 ms) later. This indicates that the prediction-based term has caused the calculated margin of safety to be expanded resulting in the safety filter needing to intervene later thus giving the main controller more freedom. So to summarize, the PB-CBF construction not only allows for a CBF to operate feasibly in an input-constrained system but also reduces the base CBF’s conservativeness when it is not necessary.
\begin{figure}[h]
\centering
\centerline{\includegraphics[width=\columnwidth]{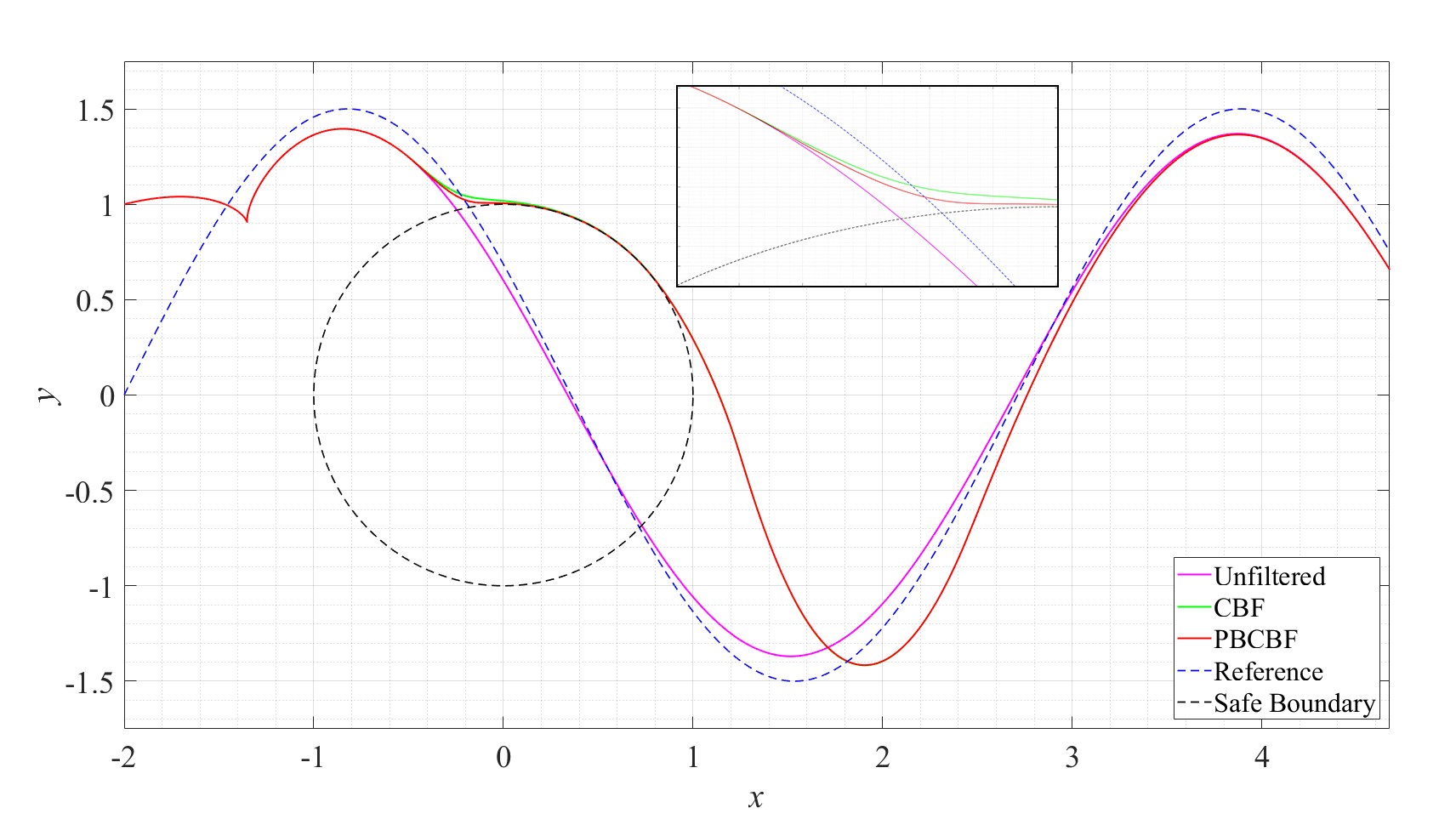}}
(a)
\centerline{\includegraphics[width=\columnwidth]{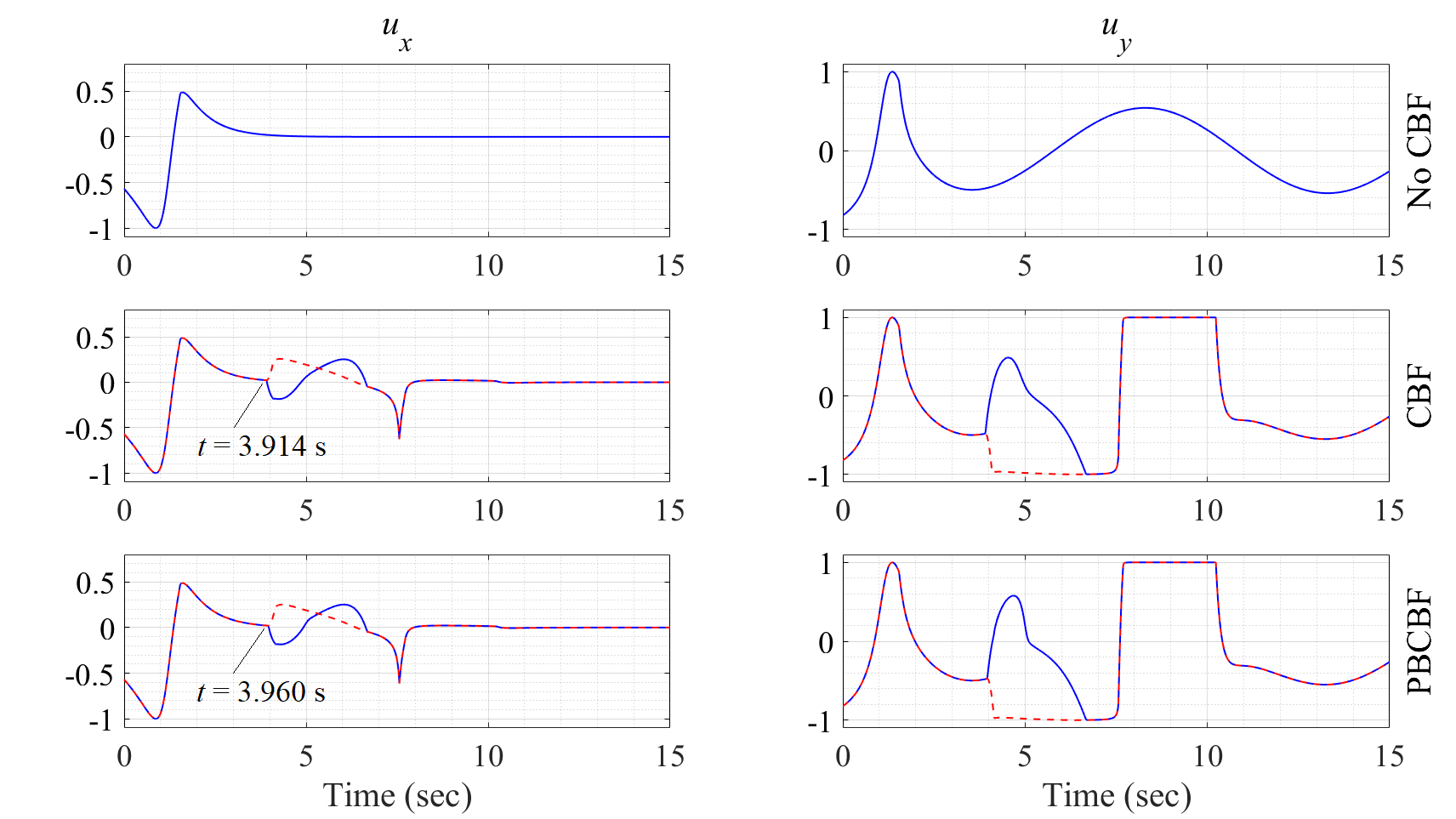}}
(b)
\caption{Simulation results for the double integrator system: (a) the path, (b) the inputs for scenarios with (\textit{i}) no CBFs, (\textit{ii}) the base CBF, and (\textit{iii}) the PB-CBF ($ \gamma = 2 $). The pre-modified control signals are displayed in dashed red.}
\label{FigN4}
\end{figure}
\begin{figure}[h]
\centering
\centerline{\includegraphics[width=\columnwidth]{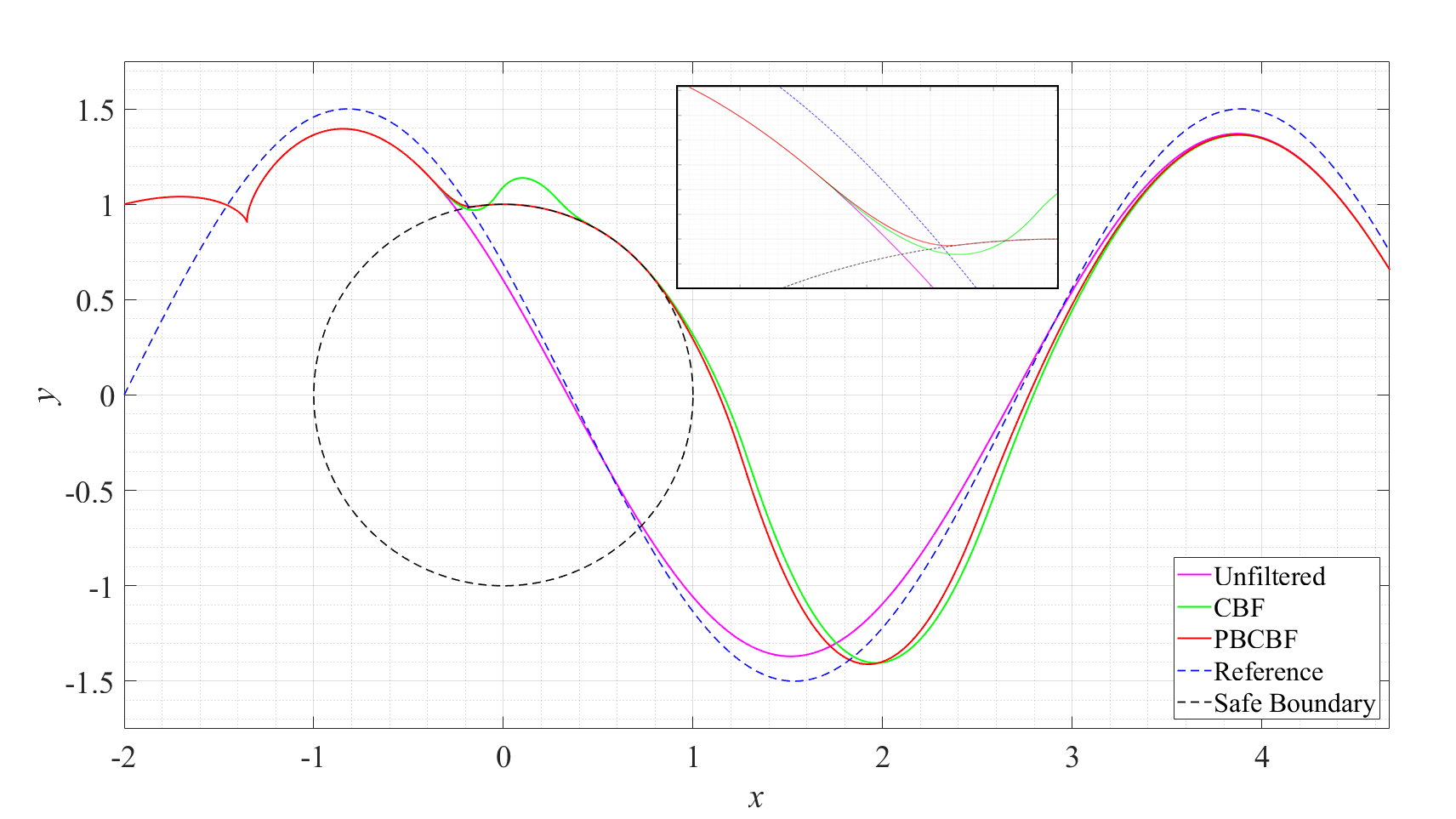}}
(a)
\centerline{\includegraphics[width=\columnwidth]{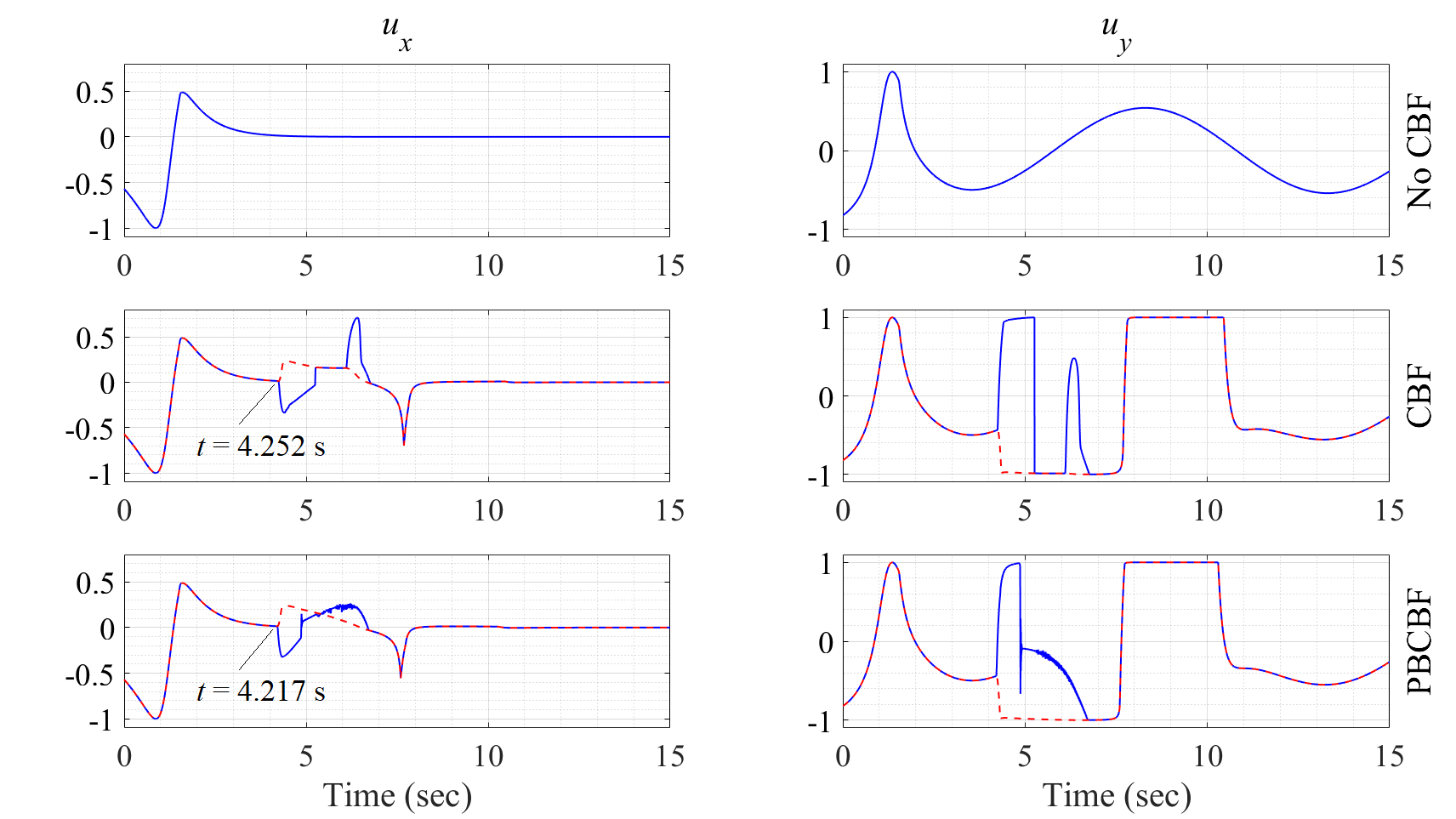}}
(b)
\caption{Simulation results for the double integrator system: (a) the path, (b) the inputs for scenarios with (\textit{i}) no CBFs, (\textit{ii}) the base CBF, and (\textit{iii}) the PB-CBF ($ \gamma = 10 $). The pre-modified control signals are displayed in dashed red.}
\label{FigN5}
\end{figure}

The first example was meant (through its simplicity) to highlight some of the advantages of the PB-CBF construction. The next example which is a bit more complicated, demonstrates the implementation of the results of the present study in a more practical engineering problem in a fairly straightforward manner.

\subsection{Airplane Stall Prevention} \label{SecIV-B}

When an airplane enters into the stall regime, a severe loss of controllability along with highly unpredictable nonlinear flight performance is usually encountered. As such, preventing stalls is a matter of utmost criticality for most transport aircraft. To this end, the aircraft’s angle of attack (AoA) with respect to the incoming airflow must be closely moderated. If AoA exceeds a specified linear range that is often dependent on the airplane’s airfoil and wing design, the airflow around the wing and other lift-producing control surfaces could detach, resulting in a sudden loss of lift forces that usually lead to stalling. Consequently, the airplane’s AoA falls neatly into the CBF framework as a safety-critical parameter.

\subsubsection{Flight Dynamics Model} 
In accordance with the notation defined in \cite{O23}, the dynamics equations of a rigid-body aircraft are:
\begin{equation}
\sum{\boldsymbol{f}_{i}^{B}}={{m}^{B}}{{D}^{B}}\boldsymbol{v}_{B}^{I}+{{m}^{B}}{{\mathbf{\Omega}}^{BI}}\boldsymbol{v}_{B}^{I}
\label{EqO40}
\end{equation}
\begin{equation}
\sum{\boldsymbol{m}_{Bi}^{B}}=\mathbf{I}_{B}^{B}{{D}^{B}}{{\boldsymbol{\omega}}^{BI}}+{{\mathbf{\Omega }}^{BI}}\mathbf{I}_{B}^{B}{{\boldsymbol{\omega}}^{BI}}
\label{EqO41}
\end{equation}
where $ \boldsymbol{f}_{i}^{B} $  and $ \boldsymbol{m}_{Bi}^{B} $  represent the forces and moments (about the center of mass) acting on the aircraft respectively, $ m^B $ is its mass, $ \mathbf{I}_{B}^{B} $  its moment of inertia about the center of mass, $ \boldsymbol{v}_{B}^{I} $  and $ \boldsymbol{\omega}^{BI} $ are its linear and angular velocity with respect to the inertial frame, $ D^B(\:\cdot\:) $ is the rotational time derivative, and $ [\:\cdot\:]^B $ will be defined as the representation of a vector in the body coordinate system $ \mathit{B} $. Subsequently, representing these equations in the aircraft’s body coordinate system and simplifying them for three degrees of freedom gives:
\begin{equation}
\begin{matrix}
\left[ \begin{matrix}
{\dot{U}}  \\
{\dot{W}}  \\
\end{matrix} \right]={{m}^{-1}}\sum{{{[f_{i}^{B}]}^{B}}}-\left[ \begin{matrix}
QW  \\
-QU  \\
\end{matrix} \right]  \\
\dot{Q}={{I}_{y}}^{-1}\sum{m_{Bi}^{B}}  \\
\dot{\theta }=Q  \\
\end{matrix}
\label{EqO42}
\end{equation}
Here, $ U $ and $ W $ are the aircraft's velocity along its $ x $ and $ z $ axis respectively, $ Q $ is its angular speed about the $ y $ axis, $ \theta $ is the aircraft pitch angle with respect to the horizon, and $ I_y $ is its moment of inertia with respect to the body $ y $ axis. Furthermore, the aircraft’s angle of attack $ \alpha $ (not to be mistaken with the class $ \text{K}_\infty$ functions for which the same symbol was used), is given by \eqref{EqO43}.
\begin{equation}
\tan \alpha =\frac{W}{U}
\label{EqO43}
\end{equation}
Additionally, \eqref{EqO43} can be differentiated with respect to time to give \eqref{EqO44} ($ \mathbf{v} = [U\;W]^\text{T} $).
\eqref{EqO43}.
\begin{equation}
\dot{\alpha }=\frac{\dot{W}U-W\dot{U}}{{{U}^{2}}+{{W}^{2}}}=\frac{[\begin{matrix}
-W & U  \\
\end{matrix}]}{{{U}^{2}}+{{W}^{2}}}\left[ \begin{matrix}
{\dot{U}}  \\
{\dot{W}}  \\
\end{matrix} \right]={{\mathbf{{c}'}}^{\text{T}}}\dot{\mathbf{v}}
\label{EqO44}
\end{equation}

The force and moment terms seen in \eqref{EqO42} can also be expanded into the following form:
\begin{equation}
\sum{{{[f_{i}^{B}]}^{B}}}={{[{{f}_{a}}]}^{B}}+{{[{{f}_{t}}]}^{B}}+m{{[g]}^{B}}
\label{EqO45}
\end{equation}
\begin{equation}
\sum{m_{Bi}^{B}}={{m}_{a}}
\label{EqO46}
\end{equation}
Defining $ \boldsymbol{\zeta}  = [a\;\bar{\text{M}}\;Q]^\text{T} $, where $ \bar{\text{M}} =  \|\mathbf{v}\|_2/a $  represents the Mach number and $ a $ the speed of sound, the aerodynamic force $ [f_a]^B $ and moment $ m_a $ are given by:
\begin{equation}
{{[{{f}_{a}}]}^{B}}=-\bar{q}S{{[T]}^{BS}}({\boldsymbol{c}_{F0}}+{\boldsymbol{C}_{Fi}}\zeta +{\boldsymbol{c}_{F\dot{\alpha}}}{{\mathbf{{c}'}}^{\text{T}}}\mathbf{\dot{v}}+{\boldsymbol{c}_{F\delta E}}{{\delta }_{E}})
\label{EqO47}
\end{equation}
\begin{equation}
{{m}_{a}}=\bar{q}Sc[{{C}_{m0}}+\boldsymbol{c}_{mi}^{\text{T}}\zeta +{{C}_{m\delta E}}{{\delta }_{E}}+{{C}_{m\dot{\alpha}}}{{\mathbf{{c}'}}^{\text{T}}}\mathbf{\dot{v}}] \;.
\label{EqO48}
\end{equation}
Here, $ \delta_E \in [\delta_{E\min},\delta_{E\max}] $ is the elevator angle input, and $ \boldsymbol{c}_{F0} $, $ \boldsymbol{C}_{Fi} $, $ \boldsymbol{c}_{F\dot{\alpha}} $, $ \boldsymbol{c}_{F\delta E} $ and $ \boldsymbol{c}_{mi}^{\text{T}} $  are matrices containing the aerodynamic coefficients $ C_i $ (these coefficients along with their corresponding formulas can be found in \cite{O24} and \cite{O25}):
\begin{equation}
\begin{matrix}
\begin{matrix}
{\boldsymbol{c}_{F0}}={{[\begin{matrix}
{{C}_{L0}} & {{C}_{D0}}  \\
\end{matrix}]}^{\text{T}}} & , & {\boldsymbol{C}_{Fi}}=\left[ \begin{array}{*{35}{l}}
{{C}_{L\alpha }} & {{C}_{L\bar{\text{M}}}} & {{C}_{Lq}}  \\
{{C}_{D\alpha }} & {{C}_{D\bar{\text{M}}}} & {{C}_{Dq}}  \\
\end{array} \right]  \\
\end{matrix}  \\
\begin{matrix}
{\boldsymbol{c}_{F\dot{\alpha }}}={{[\begin{matrix}
{{C}_{L\dot{\alpha }}} & {{C}_{D\dot{\alpha }}}  \\
\end{matrix}]}^{\text{T}}} & , & {\boldsymbol{c}_{F\delta E}}={{[\begin{matrix}
{{C}_{L\delta E}} & {{C}_{D\delta E}}  \\
\end{matrix}]}^{\text{T}}}  \\
\end{matrix}  \\
\boldsymbol{c}_{mi}^{\text{T}}=[\begin{matrix}
{{C}_{m\alpha }} & {{C}_{m\bar{\text{M}}}} & {{C}_{mq}}  \\
\end{matrix}]  \\
\end{matrix}
\label{EqO49}
\end{equation}
and $ [T]^{BS} $ is the transformation matrix from the stability to the body coordinate system.
\begin{equation}
{{[T]}^{BS}}=\left[ \begin{matrix}
\cos \alpha  & -\sin \alpha   \\
\sin \alpha  & \cos \alpha   \\
\end{matrix} \right]
\label{EqO51}
\end{equation}
The thrust force $ [f_t]^B $ is given by \eqref{EqO52} as follows:
\begin{equation}
{{[{{f}_{t}}]}^{B}}={{[\begin{matrix}
m{{X}_{\delta th}}{{\delta }_{th}} & 0  \\
\end{matrix}]}^{\text{T}}}=m{{X}_{\delta th}}{{\delta }_{th}}\mathbf{i}
\label{EqO52}
\end{equation}
Where $ \delta_{th} \in [0,100] $ is the throttle input and $ X_{\delta th} $ can be found according to the number of engines $ n_\text{Engine} $ and maximum engine thrust $ T_{{\max},\text{Engine}} $ as given by \eqref{EqO53}.
\begin{equation}
{{X}_{\delta th}}=\frac{{{n}_{\text{Engine}}}{{T}_{\max ,\text{Engine}}}}{100m}
\label{EqO53}
\end{equation}
Finally, the gravitational force $ m[g]^B $ can also be written according to \eqref{EqO54}.
\begin{equation}
{{[g]}^{B}}={{[T]}^{BI}}{{[\begin{matrix}
0 & g  \\
\end{matrix}]}^{\text{T}}}=g{{[T]}^{BI}}\mathbf{k}
\label{EqO54}
\end{equation}
Where $ g $ is the local gravitational acceleration and $ [T]^{BI} $ is the inertial to the body transformation matrix.
\begin{equation}
{{[T]}^{BI}}=\left[ \begin{matrix}
\cos \theta  & -\sin \theta   \\
\sin \theta  & \cos \theta   \\
\end{matrix} \right]
\label{EqO55}
\end{equation}

To represent the dynamical equations in the affine state-space form, we shall define:
\begin{align}
\begin{split}
{{\mathbf{f}}_{\text{T}}}=g{{[T]}^{BI}}\mathbf{k}& -[\begin{matrix}
QW & -QU{{]}^{\text{T}}}
\end{matrix}\\
&-{{m}^{-1}}\bar{q}S{{[T]}^{BS}}({\boldsymbol{c}_{F0}}+{\boldsymbol{C}_{Fi}}\boldsymbol{\zeta} )
\end{split}
\label{EqO56}
\end{align}
\begin{equation}
{{\mathbf{G}}_{\text{T}}}=[\begin{matrix}
-{{m}^{-1}}\bar{q}S{{[T]}^{BS}}{\boldsymbol{c}_{F\delta E}} & {{X}_{\delta th}}\mathbf{i}  \\
\end{matrix}]
\label{EqO57}
\end{equation}
\begin{equation}
{{\mathbf{A}}_{\text{T}}}=\mathbf{I}+{{m}^{-1}}\bar{q}S{{[T]}^{BS}}{\boldsymbol{c}_{F\dot{\alpha }}}{{\mathbf{{c}'}}^{\text{T}}}
\label{EqO58}
\end{equation}
\begin{equation}
{{\text{f}}_{\text{R}}}={{I}_{y}}^{-1}\bar{q}Sc({{C}_{m0}}+\boldsymbol{c}_{mi}^{\text{T}}\boldsymbol{\zeta})
\label{EqO59}
\end{equation}
\begin{equation}
\mathbf{g}_{\text{R}}^{\text{T}}=[\begin{matrix}
{{I}_{y}}^{-1}\bar{q}Sc{{C}_{m\delta E}} & 0  \\
\end{matrix}]
\label{EqO60}
\end{equation}
\begin{equation}
\mathbf{a}_{\text{R}}^{\text{T}}=-{{I}_{y}}^{-1}\bar{q}Sc{{C}_{m\dot{\alpha }}}{{\mathbf{{c}'}}^{\text{T}}} \;.
\label{EqO61}
\end{equation}
So, defining the state and input vectors as $ \mathbf{x} = [U\;W\;Q\;\theta]^\text{T} $ and $ \mathbf{u} = [\delta_E\;\delta_{th}]^\text{T} $ correspondingly, the equations fall into the \eqref{EqO3} form, with $ \mathbf{f} $ and $ \mathbf{G} $ given by \eqref{EqO62} and \eqref{EqO63}.
\begin{equation}
\mathbf{f}(\mathbf{x})={{\left[ \begin{matrix}
{{\mathbf{A}}_{\text{T}}} & \mathbf{0} & \mathbf{0}  \\
\mathbf{a}_{\text{R}}^{\text{T}} & 1 & 0  \\
\mathbf{0} & 0 & 1  \\
\end{matrix} \right]}^{-1}}\left[ \begin{matrix}
{{\mathbf{f}}_{\text{T}}}  \\
{{\text{f}}_{\text{R}}}  \\
Q  \\
\end{matrix} \right]
\label{EqO62}
\end{equation}
\begin{equation}
\mathbf{G}(\mathbf{x})={{\left[ \begin{matrix}
{{\mathbf{A}}_{\text{T}}} & \mathbf{0} & \mathbf{0}  \\
\mathbf{a}_{\text{R}}^{\text{T}} & 1 & 0  \\
\mathbf{0} & 0 & 1  \\
\end{matrix} \right]}^{-1}}\left[ \begin{matrix}
{{\mathbf{G}}_{\text{T}}}  \\
\mathbf{g}_{\text{R}}^{\text{T}}  \\
0  \\
\end{matrix} \right]
\label{EqO63}
\end{equation}

The equations represented in the affine form would allow for the implementation of PB-CBF in the prediction-based QP framework.

\subsubsection{Control System and Control Barrier Function}

As a next step, a controller that would serve to stabilize the aircraft in cruising flight shall be first introduced. In this respect and for the sake of simplicity, a PID stability augmentation system (SAS) and auto-throttle described by \eqref{EqO64}, was selected to maintain the preset trim conditions in cruising flight.
\begin{align}
\begin{split}
&\mathbf{u}(t)={{\mathbf{u}}_{\text{Trim}}}\\
&+{{\mathcal{L}}^{-1}}\left[ \begin{matrix}
\frac{{{k}_{P\alpha }}}{{{U}_{0}}}(W-{{W}_{0}})+{{k}_{D\theta }}Q+\left( {{k}_{P\theta }}+\frac{{{k}_{I\theta }}}{s} \right)(\theta -{{\theta }_{0}})  \\
\left( {{k}_{Pu}}+\frac{{{k}_{Iu}}}{s}+{{k}_{Du}}s \right)(U-{{U}_{0}})  \\
\end{matrix} \right]
\end{split}
\label{EqO64}
\end{align}

Where $ \mathbf{x}_0 = [U_0\;W_0\;0\;\theta_0]^\text{T} $ is the trim state, $ \mathbf{u}_\text{Trim} $ is the trim input vector, and the $ k_i $'s are constant control design parameters. As mentioned earlier, if preventing an aerodynamic stall is the goal in mind, the angle of attack $ \alpha(\mathbf{x}) $ will be the safety-critical parameter, which is to be kept inside the range $ \alpha_{\min} < \alpha < \alpha_{\max} $. With this definition, $ \mathbf{c}^\text{T} $ from \eqref{EqO17} will be given by \eqref{EqO65}.
\begin{equation}
{{\mathbf{c}}^{\text{T}}}(\mathbf{x})=[\begin{matrix}
{{{\mathbf{{c}'}}}^{\text{T}}} & 0 & 0  \\
\end{matrix}]
\label{EqO65}
\end{equation}
Furthermore, since the airplane’s controls are expected to approximately behave in a linear fashion as long as near-trim conditions are maintained and the no-stall assumption holds, we can simplify the prediction process by using the linear approximation of \eqref{EqO30}; and so, an appropriate selection of $ \mathbf{x}_0 $ would be the trim state, for which we shall have:
\begin{equation}
\mathbf{A}=\left[ \begin{matrix}
{{X}_{u}} & {{X}_{w}} & 0 & -g\cos {{\theta }_{0}}  \\
{{X}_{u}} & {{Z}_{w}} & {{U}_{0}} & -g\sin {{\theta }_{0}}  \\
{{{\bar{M}}}_{u}} & {{{\bar{M}}}_{w}} & {{{\bar{M}}}_{q}} & 0  \\
0 & 0 & 1 & 0  \\
\end{matrix} \right]
\label{EqO66}
\end{equation}
\begin{equation}
\mathbf{B}={{\left[ \begin{array}{*{35}{l}}
{{X}_{\delta E}} & {{Z}_{\delta E}} & {{{\bar{M}}}_{\delta E}} & 0  \\
{{X}_{\delta th}} & {{Z}_{\delta th}} & {{{\bar{M}}}_{\delta th}} & 0  \\
\end{array} \right]}^{\text{T}}}
\label{EqO67}
\end{equation}
\begin{equation}
{{\mathbf{c}}^{\text{T}}}=[\begin{matrix}
0 & {{U}_{0}}^{-1} & 0 & 0  \\
\end{matrix}]
\label{EqO68}
\end{equation}
The parameters $ X_i $, $ Z_i $ and  $ \bar{M}_i $  are known as the stability derivatives, which can be found from the aerodynamic coefficients and other airplane parameters that have been previously introduced. The equations for obtaining the stability derivatives from the aerodynamic coefficients based on the trim condition can be found in \cite{O24}, \cite{O25} and \cite{O26}.

According to the stall prevention requirements $ \alpha_{\min} < \alpha(\mathbf{x}) < \alpha_{\max} $; so, based on the previous assertions, the CBFs
\begin{equation}
{{h}_{1}}(\mathbf{x})={{\alpha }_{\max }}-\alpha (\mathbf{x})
\label{EqO69}
\end{equation}
\begin{equation}
{{h}_{2}}(\mathbf{x})=\alpha (\mathbf{x})-{{\alpha }_{\min }}
\label{EqO70}
\end{equation}
with$ \mathbf{c}_1 = -\mathbf{c} $ and $ \mathbf{c}_2 = \mathbf{c} $ could be used to adjust the PID signal controller so that the angle of attack is kept inside the desired interval. Notably, it can be shown that the corresponding entry of $ \mathbf{c}^\text{T}\mathbf{A}\mathbf{B} $ for the throttle input is always zero if $ \mathbf{c}^\text{T} $, $ \mathbf{A} $ and $ \mathbf{B} $ take the forms given by \eqref{EqO66}, \eqref{EqO67}, and \eqref{EqO68}, meaning that it will be left at its initial value in the predictions. This would happen when an input has no effect in changing the safety-critical variable (according to which the base CBF is constructed). This is approximately true for the throttle input which does work to bring the angle of attack to zero, but it does so at a substantially slower and less effective rate than the elevator input command. That is why the linear approximation has neglected this effect. Fig.~\ref{FigO3} shows the complete control diagram. Before discussing the simulation results, though not the primary concern here, it may be useful to briefly address the justification for the added complexity of incorporating a CBF in the first place.
\begin{figure}[h]
\centerline{\includegraphics[width=\columnwidth]{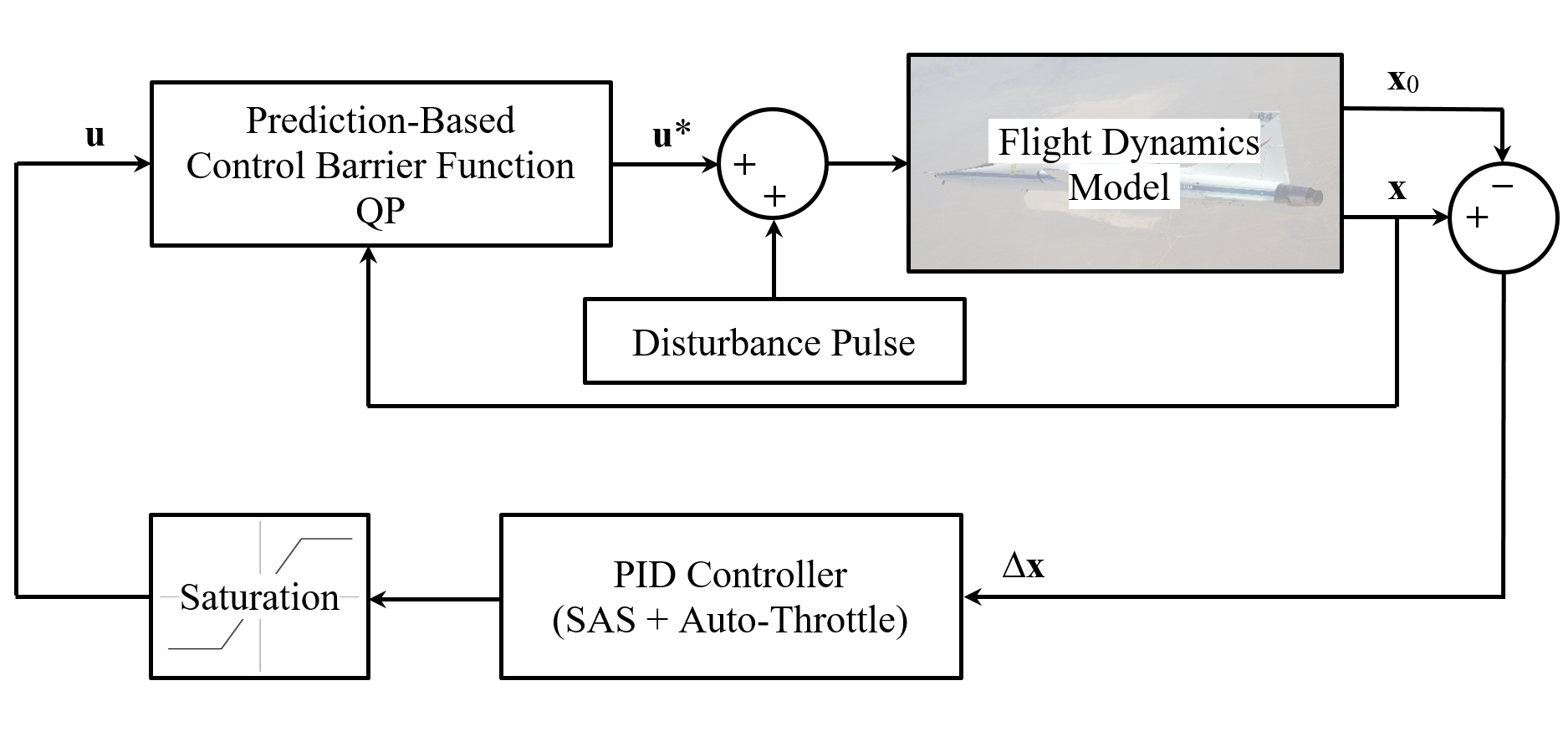}}
\caption{The control block diagram for the flight dynamics example.}
\label{FigO3}
\end{figure}

If the control system is sufficiently well-designed, it should not only be able to stabilize the system but also regulate secondary variables such as the angle of attack to prevent unsafe behavior. In regard to this specific query, it should be noted that many popular control schemes like PID controllers, pole placement-based state feedback controllers, or the linear optimal infinite-horizon LQR, are designed without any direct considerations of safety. That is to say, for such controllers, the alternative ways to show safety would be through extensive tests and simulations, or through the introduction of other nonlinear elements into the control system. Nonetheless, the former would merely provide one with an assurance instead of a formal guarantee, while the latter would most likely be not any less complex than a CBF if it is to provide the same level of guarantees. Furthermore, although CBFs are a source of added complexity, their clever selection would allow for a guaranteed avoidance of certain nonlinearities in the system that can be notoriously hard to control in most complex aeronautical problems, in turn allowing for less complexity in other elements of the control system. Therefore, having provided some justification for the utility of PB-CBFs in an aeronautical application, we shall now proceed to present and analyze the simulation results.

\subsubsection{Simulation Results}

The numerical simulations were done using the aerodynamic data given by \cite{O24}, for a high-performance training jet flying at subsonic velocities near sea level altitudes. The aircraft is cruising at a speed of $ \|v0\|_2 = 85.34\;\text{m.s}^{-1}$ and a nominal AoA of $ \alpha_0 \approx 10^\circ $. Furthermore, the elevator input is limited to $ \delta_E \in [-20^\circ, 10^\circ] $ and the angle of attack is to be kept inside the range $ \alpha \in [-10^\circ, 15^\circ] $.

The simulation starts with the aircraft in trim state, a doublet pulse disturbance of $ \Delta \mathbf{u}_D = (10^\circ)\Pi(t / 5 - 0.5) + (-20^\circ)\Pi(t / 5 - 1.5) $ is then introduced to the system as shown in Fig.~\ref{FigO3}. The class $ \text{K}_\infty$ function $ h (z) = \gamma z $ (a positive constant) was chosen with an equal value in all CBFs for the sake of comparison. Simulation results including the control inputs and the angle of attack graphs for three cases can be seen in Fig.~\ref{FigO4}.
\begin{figure}[h]
\centering
\centerline{\includegraphics[width=\columnwidth]{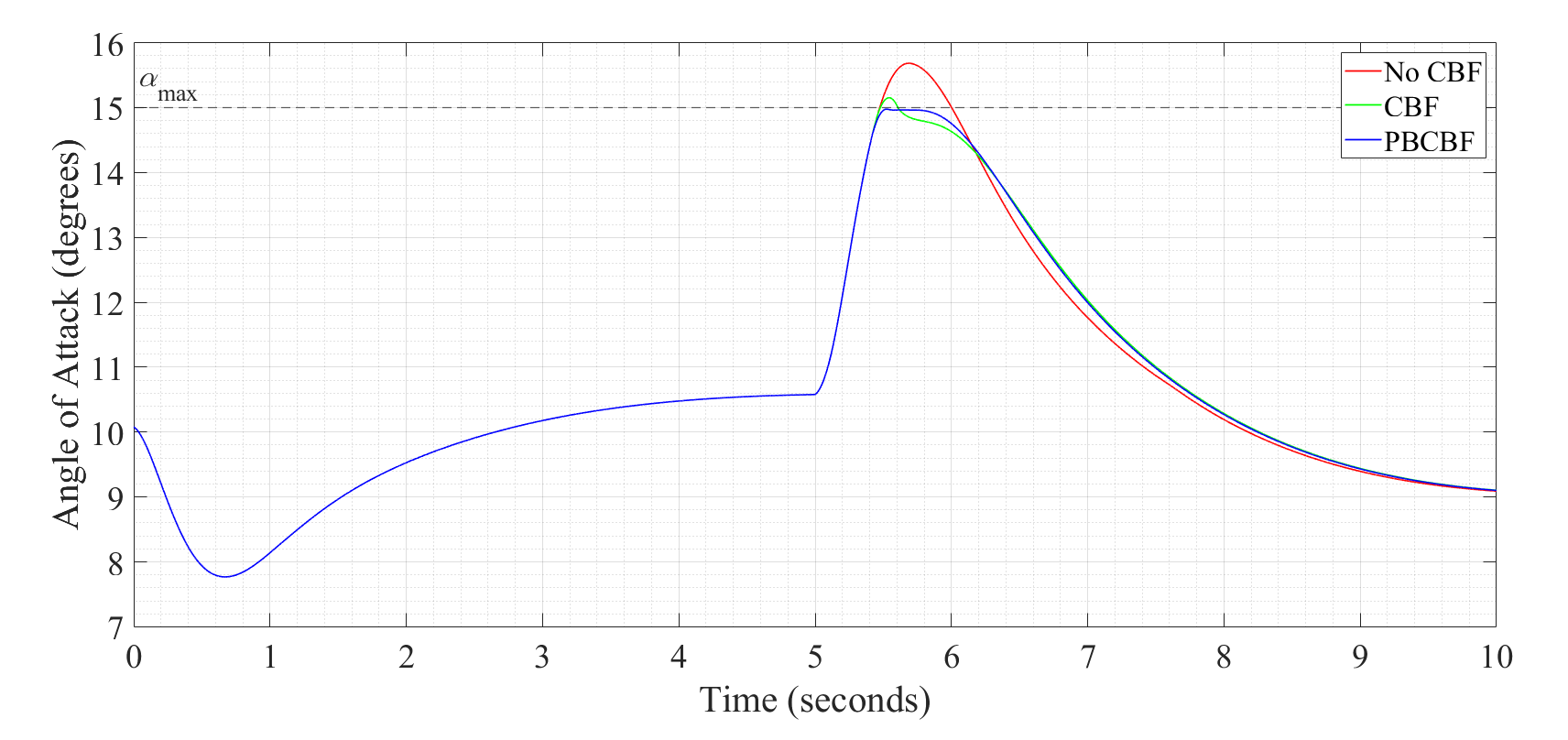}}
(a)
\centerline{\includegraphics[width=\columnwidth]{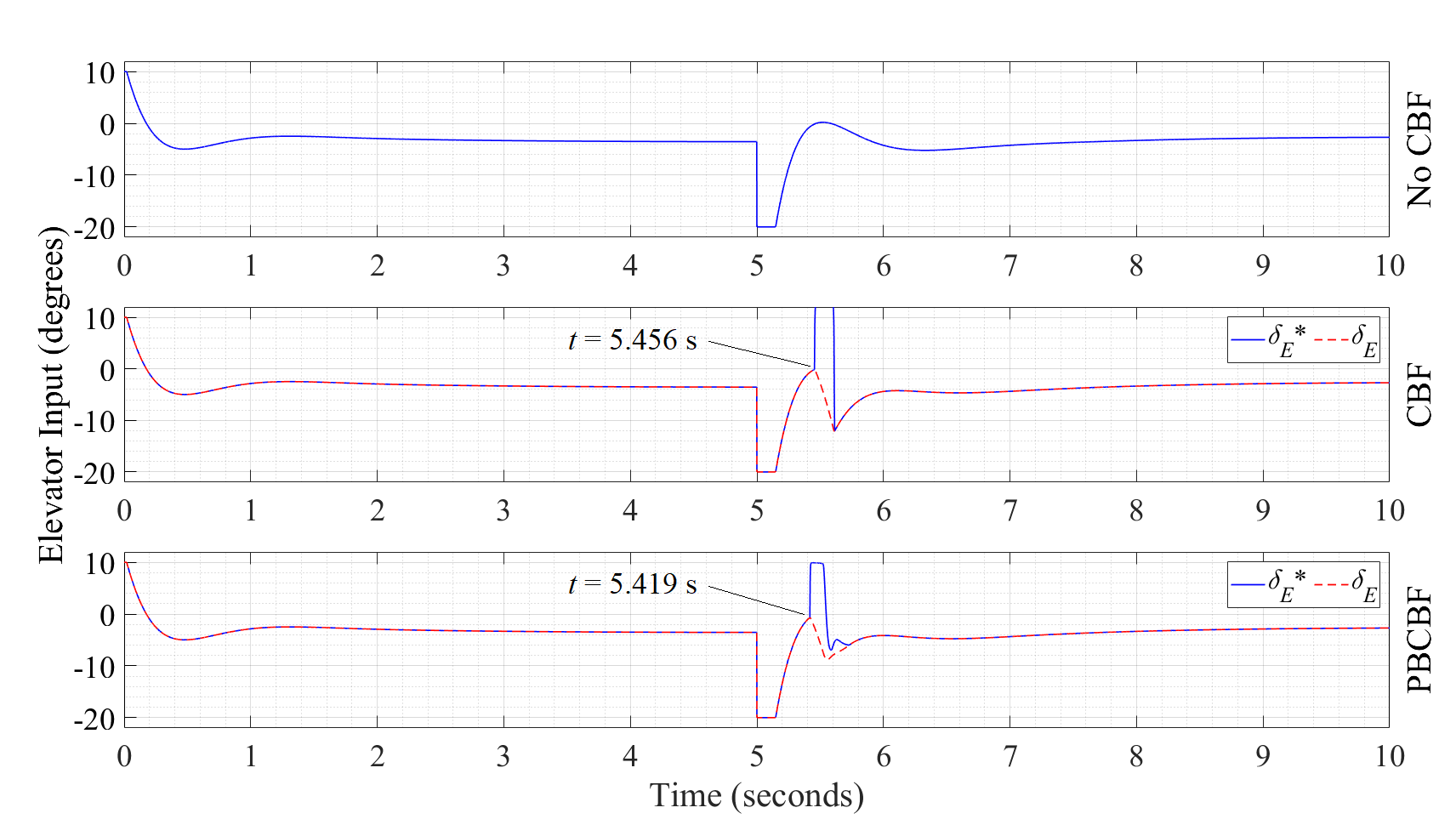}}
(b)
\centerline{\includegraphics[width=\columnwidth]{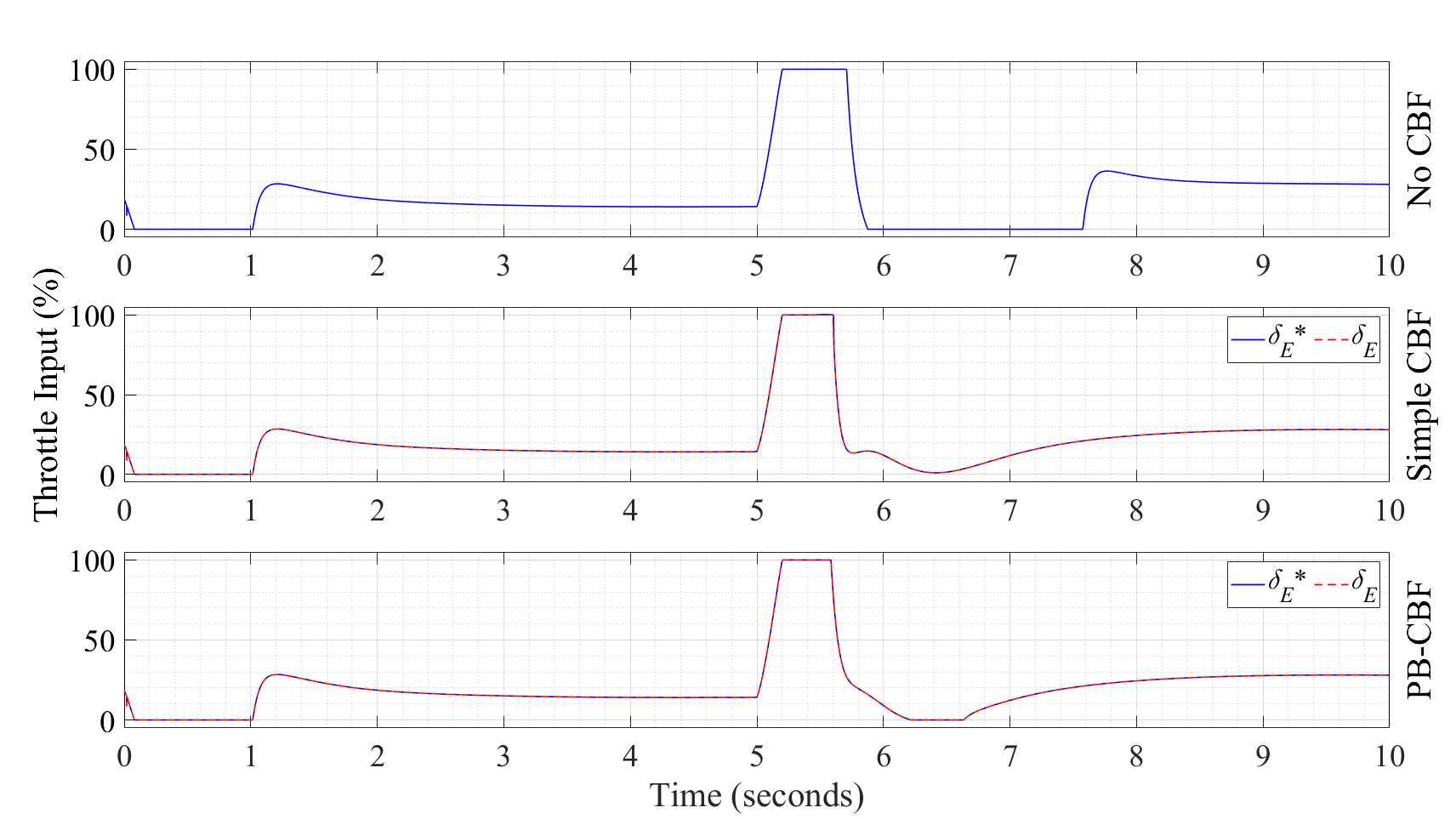}}
(c)
\caption{Simulation results for the stall prevention system: (a) the angle of attack, (b) the elevator input, and (c) the throttle input for three cases where: (\textit{i}) no CBFs are used, (\textit{ii}) the base CBFs are used (which do not use prediction), and (\textit{iii}) the PB-CBF scheme is implemented. The pre-modified control signals are displayed in dashed red.}
\label{FigO4}
\end{figure}

There are a few points to note from Fig.~\ref{FigO4}. Firstly, it is demonstrated that without the usage of CBFs, the disturbance would result in a firm crossing of the safe boundary for the AoA, which would result in a stall and subsequent degradation of the flight dynamic behavior if aerodynamic stalls have been modeled. Additionally, although the simple CBF scheme would make adjustments to the control signal as the safe boundary is approached, due to the control constraints not being considered explicitly, the modifications are not completely satisfying, leading to exceeding AoA set limits and control saturation (see Fig.~\ref{FigO4}-b). Secondly, it is noticed that the introduction of a prediction-based term into the CBF has resolved this issue by expediting the response by a few milliseconds resulting in feasible modifications. Needless to say, this activation window would likely be slightly longer for slower systems with more restrictive control constraints, but it nonetheless remains as a means of intervention that only interferes if safety is at risk. Thirdly and finally, as per the previous statements and expectations, it can be seen that neither of the CBFs has made any substantial modifications to the throttle input, which is due to its negligible effect on the safety parameter.

\section{Conclusion and Remarks} \label{SecV}

The present study has introduced a new prediction-based approach to control barrier functions (CBFs) aimed at input-constrained problems. The proposed framework and theory for prediction-based control barrier functions (PB-CBFs) works by propagating the system model from the present state with a control law aimed at stopping advances towards the boundaries of the safe set in a (near) optimal fashion. The desired functionality is therefore achieved by calculating the margin needed to stop given the constraints on the inputs, and then subtracting the calculated margin from the original CBF to allow it to be feasible with input constraints. While calculating the margin may be analytically possible in some cases, the present study introduced a simulation-based method that may be broadly utilized. This functionality was clearly demonstrated via the presented numerical examples where it was shown that the proposed PB-CBFs cause the activation of the safety filter to be sufficiently expedited when needed in order to prevent the crossing of safe boundaries. It was also shown that PB-CBFs can be less disruptive than the base CBFs when the available margin is large. Furthermore, as illustrated by the first example, another notable feature of PB-CBFs is their asymmetric trimming of the safe set, allowing high energy states that pose no threat to safety to be kept thereby being less restrictive on the main controller.

In propagating the model, the near-optimal control laws used in the present study were \textit{time-optimal}. That is to say, they minimize the time needed to stop the advance toward the safe boundary; this is instead of minimizing the total reduction in safety itself. Future work could investigate the possibility of doing the latter instead of the former to see its advantages and potential computational disadvantages. Another interesting directive will be to look into how PB-CBFs need to be augmented to work in the presence of unknown disturbances or uncertainties. Since PB-CBFs rely on a deterministic propagation of the model to find the stop margin, the introduction of a persistent and unknown disturbance may result in an incorrect estimation of the stop margin; hence, finding potential remedies to this particular issue will be of practical value.

\bibliographystyle{ieeetr}
\bibliography{references}

\end{document}